\documentclass[12pt]{article}
\usepackage{a4wide}
\usepackage{amssymb}
\usepackage{amsmath}
\usepackage{mathrsfs}

\begin{document}
{\renewcommand{\thefootnote}{\fnsymbol{footnote}}
\begin{center}
{\LARGE  Quantization of dynamical symplectic reduction}\\
\vspace{1.5em}
Martin Bojowald$^1$\footnote{e-mail address: {\tt bojowald@gravity.psu.edu}}
and Artur Tsobanjan$^2$\footnote{e-mail
    address: {\tt arturtsobanjan@kings.edu}} 
\\
\vspace{0.5em}
$^1$Institute for Gravitation and the Cosmos,
The Pennsylvania State
University,\\
104 Davey Lab, University Park, PA 16802, USA\\
\vspace{0.5em}
 $^2$ King's College, 133 North River Street, Wilkes--Barre, PA 18711, USA\\ 
  \vspace{1.5em}
\end{center}
}

\setcounter{footnote}{0}

\newtheorem{theo}{Theorem}
\newtheorem{prop}{Proposition}
\newtheorem{cor}{Corollary}
\newtheorem{lemma}{Lemma}
\newtheorem{defi}{Definition}

\newcommand{\proofend}{\raisebox{1.3mm}{\fbox{\begin{minipage}[b][0cm][b]{0cm}
\end{minipage}}}}
\newenvironment{proof}{\noindent{\it Proof:} }{\mbox{}\hfill \proofend\\\mbox{}}
\newenvironment{ex}{\mbox{}\\\noindent{\it Example:} }{\mbox{}\\\smallskip}
\newenvironment{rem}{\mbox{}\\\noindent{\it Remark:} }{\mbox{}\\\smallskip}

\newcommand{\case}[2]{{\textstyle \frac{#1}{#2}}}
\newcommand{\lP}{\ell_{\mathrm P}}

\newcommand{\md}{{\mathrm{d}}}
\newcommand{\tr}{\mathop{\mathrm{tr}}}
\newcommand{\sgn}{\mathop{\mathrm{sgn}}}

\newcommand*{\R}{{\mathbb R}}
\newcommand*{\N}{{\mathbb N}}
\newcommand*{\Z}{{\mathbb Z}}
\newcommand*{\Q}{{\mathbb Q}}
\newcommand*{\C}{{\mathbb C}}

\begin{abstract}
  A long-standing problem in quantum gravity and cosmology is the quantization
  of systems in which time evolution is generated by a constraint that must
  vanish on solutions. Here, an algebraic formulation of this problem is
  presented, together with new structures and results, which show that
  specific conditions need to be satisfied in order for well-defined evolution
  to be possible.
\end{abstract}

\section{Introduction}
\label{s:Intro}
When time-reparameterization-invariant dynamical systems are cast as
Hamiltonian theories on a symplectic manifold one finds that time evolution
and time-reparameterization flows are generated by one and the same
phase-space function---the Hamiltonian constraint. The straightforward
application of the usual methods of symplectic reduction to such dynamically
constrained systems has the undesirable side-effect of also removing their
dynamics, and needs to be replaced by \emph{dynamical} syplectic
reduction. This paper describes a method of dynamical reduction for the
quantized versions of such systems, where non-commutativity leads to a host of
additional complications. However, since dynamically constrained systems are
rarely studied outside of canonical approaches to quantum gravity we dedicate
most of this introductory section to the review of their classical
(i.e.\ non-quantum) treatment. Our main results and the structure of the rest
of this manuscript are outlined at the end of the introduction.

Given a symplectic manifold $(M,\Omega)$ and $C\in C^{\infty}(M)$, the {\em
  symplectic reduction} \cite{SympRed} $M/C$ of $M$ by $C$ is the orbit space
of $M\supset M_C\colon C=0$ with respect to the gauge flow
$F_C(\epsilon)=\exp(\epsilon X_C)$ in $M_C$ generated by the Hamiltonian
vector field $X_C$ of $C$, ${\rm d}C=\Omega(X_C,\cdot)$.  Because ${\cal
  L}_{X_C}C=\Omega(X_C,X_C)=0$, the flow preserves $M_C$, and the orbit space
inherits a unique symplectic form $\Omega_{M/C}$ from the presymplectic form
$i^*\Omega$ on $M_C$, where $i\colon M_C\to M$ is the inclusion of $M_C$ in
$M$. The set of {\em observables} of the constrained system, which solve the
constraint equation $C=0$ and are invariant under the gauge flow, is given by
$C^{\infty}(M/C)$.

In addition to implementing a constraint $C=0$ by symplectic reduction,
physical systems usually require the definition of a {\em dynamical flow}. The
canonical way is to select a Hamiltonian function $H\in C^{\infty}(M)$ which
generates the dynamical flow $F_H(t)=\exp(t X_H)$ with the Hamiltonian vector
field $X_H$ of $H$. A dynamical flow in the presence of a constraint $C=0$ is
consistent if it preserves the constraint surface, that is,
$X_HC=\Omega(X_C,X_H)=-\{C,H\}=0$ on $M_C$ with the Poisson bracket
$\{\cdot,\cdot\}$\ defined by $\Omega$. The same condition ensures that the
dynamical flow is well-defined on the reduced phase space $M/C$ because it is
compatible with the gauge flow: By the Jacobi identity of $\{\cdot,\cdot\}$, a
gauge transformation (that is, the diffeomorphism induced by a gauge flow)
commutes with the dynamical flow up to a gauge transformation.  Since
$\{C,H\}=0$ on $M_C$, there is a $\lambda\in C^{\infty}(M)$ such that
$\{C,H\}=\lambda C$ on $M$, and
\[
 [X_C,X_H]= \{\{\cdot,H\},C\}- \{\{\cdot,C\},H\}= -\{\cdot,\{C,H\}\}=
 -X_{\{C,H\}}= -X_{\lambda C}\,.
\]

In systems typically encountered in general relativity or its cosmological
models, the dynamical flow is simultaneously a gauge flow. A system is {\em
  time-reparameterization invariant} if, given a solution $f(t)$ of its
dynamical flow such that ${\rm d}f/{\rm d}t=\{f,H\}$ for all $t\in{\mathbb
  R}$, $f(T(t))$ is also a solution for any monotonic $T\in
C^{\infty}({\mathbb R})$.  Any such $f(T(t))$\ can be obtained from $f(t)$ by
following the flow generated by the Hamiltonian itself together with a
suitable non-zero multiplier $N\in C^{\infty}({\mathbb R})$ via
\[
\lim_{\epsilon\to 0}\frac{f(t+\epsilon N(t))-f(t)}{\epsilon}= N(t)
\frac{{\rm d}f}{{\rm d}t}= \{f,N H\} \ .
\]
Therefore the Hamiltonian function is itself the generator of a gauge flow.
Observables
are functions on the orbit space of the gauge flow. This orbit space inherits
a Poisson structure from $M$, with symplectic leaves given by the level
surfaces of $H$ \cite{brackets}. Adding a constant to $H$ does not change the
dynamical flow. Therefore, without loss of generality, we can assume the
relevant symplectic leaf to be given by $H=0$, such that the dynamical
generator $H$ is also a constraint.

The Hamiltonian of a time-reparameterization invariant system is therefore a
constraint, called the Hamiltonian constraint. In order to emphasize its
nature as a constraint, we will slightly change notation and refer to a
Hamiltonian constraint as $C$. We refer to symplectic reduction with a
Hamiltonian constraint as {\em dynamical symplectic reduction}. Associated
with this process is the following long-standing problem
\cite{KucharTime,Isham:Time}: Any observable $O\in C^{\infty}(M/C)$ on the
reduced phase space can be pulled back to a function on $M_C\colon C=0$ using
the projection $p\colon M_C\to M/C$ to the orbit space. By definition, $p^*O$
is constant on the orbits, or time independent if $C$ is a Hamiltonian
constraint. In the reduced phase space, therefore, there is no recognizable
time evolution in a time-reparameterization invariant theory.

Classically, the problem of identifying time evolution in a
time-reparameterization invariant system is usually solved by {\em fixing the
  gauge flow} generated by a Hamiltonian constraint. This construction to
determine observables and their evolution does not use the reduced phase
space. Given a symplectic manifold $(M,\Omega)$ and a Hamiltonian constraint
$C\in C^{\infty}(M)$, a gauge fixing of the gauge flow is accomplished by a
global incisive section.
\begin{defi} \label{ClassSection}
  A {\em section} $(L,\Omega_L,\iota)$ of the gauge flow generated by a
  constraint $C$ on $(M,\Omega)$ is a symplectic manifold $(L,\Omega_L)$
  (called the {\em gauge-fixed phase space}) together with an embedding
  $\iota\colon L\to M_C$ such that $\Omega_L=\iota^*i^*\Omega$.

  A section $(L,\Omega_L,\iota)$ of the gauge flow generated by a constraint
  $C$ on $(M,\Omega)$ is {\em global} if for every $y\in M_C$ there is an
  $x\in L$ such that $y=F_C(\epsilon)\iota(x)$ for some $\epsilon$.

  A section $(L,\Omega_L,\iota)$ of the gauge flow generated by a constraint
  $C$ on $(M,\Omega)$ is {\em incisive} if, for all $x_1,x_2\in L$,
  $\iota(x_1)=F_C(\epsilon)\iota(x_2)$ for some $\epsilon$ implies $x_1=x_2$.
\end{defi}

The pull-back $\iota^*\colon C^{\infty}(M_C)\to C^{\infty}(L)$ maps functions
on the constraint surface $M_C$ to {\em gauge-fixed observables} on $L$. 
\begin{prop} \label{PropClass} If $(L,\Omega_L,\iota)$ is a global incisive
  section of the gauge flow of $C$ on $(M,\Omega)$, the gauge-fixed phase
  space $(L,\Omega_L)$ is symplectomorphic to the reduced phase space
  $(M/C,\Omega_{M/C})$.
\end{prop}
\begin{proof}
  Since a global incisive section intersects each gauge orbit exactly once,
  there is a bijection between $L$ and the reduced phase space. The
  symplectomorphism property can then be shown in local coordinates: Locally,
  $C$ can be used as a coordinate in a neighborhood around a given point $x\in
  M_C\subset M$. We use the gauge flow $F_C(\epsilon)\colon x\mapsto
  x_{\epsilon}\in M_C$ to introduce a second coordinate $z$ such that $z(x)=0$
  and $z(x_{\epsilon})=\epsilon$. The two functions $C$ and $z$ are
  canonically conjugate: $\{z,C\}=X_Cz= {\rm d}z/{\rm d}\epsilon=1$. By
  Darboux' theorem, there are ${\rm dim}M-2$ additional local coordinates
  $q_j$ and $p_k$, such that
\[
 \Omega_M= {\rm d}z\wedge{\rm d}C+ \sum_{j=1}^{\frac{1}{2}{\rm dim}M-1} {\rm
   d}q_j\wedge{\rm d}p_j \,.
\]
Since $0=\{q_j,z\}= \partial q_j/\partial C$ and
$0=\{p_j,z\}= \partial p_j/\partial C$, $q_j$ and $p_k$ together with $z$
define a local coordinate system on $M_C$.

On $M_C$, $i^*\Omega= \sum_{j=1}^{\frac{1}{2}{\rm dim}M-1} {\rm
  d}q_j\wedge{\rm d}p_j$ is a presymplectic form. Local intervals of gauge
orbits of $C$ are the coordinate lines of $z$. Therefore, $q_j$ and $p_k$ are
local coordinates on the reduced phase space, with symplectic form
$\Omega_{M/C}=\sum_{j=1}^{\frac{1}{2}{\rm dim}M-1} {\rm d}q_j\wedge{\rm
  d}p_j$. In order for $\iota^*i^*\Omega$ to be symplectic, any section of the
gauge flow must locally be of the form $\iota\colon y \mapsto (s(y), z(s(y)))$
with a canonical transformation $s\colon y\mapsto (q_J,p_k)$ and a smooth
function $z(q_j,p_k)$. Therefore, $\Omega_L=\iota^*i^*\Omega=
s^*\sum_{j=1}^{\frac{1}{2}{\rm dim}M-1} {\rm d}q_j\wedge{\rm
  d}p_j=s^*\Omega_{M/C}$.
\end{proof}

An incisive section $(L,\Omega_L,\iota)$\ {\em evolves} in $M$ if there is a
1-parameter family of incisive sections $(L,\Omega_L,\iota_t)$, $t\in
(t_1,t_2)\subset {\mathbb R}$, such that $\iota=\iota_{t_0}$\ for some $t_0\in
(t_1, t_2)$, and $L\times (t_1,t_2)\to {\cal U}, (y,t)\mapsto \iota_t(y)$ is a
diffeomorphism to an open submanifold ${\cal U}\subset M_C$. For each value of
$t \in (t_1, t_2)$\ the hypersurface $\iota_t (L) \subset M_C$\ plays the role
of a surface of a fixed value of time. With this structure in place, any
function $f \in C^\infty(M_C)$\ can be viewed as evolving in time along any
given gauge orbit by tracing it along the intersection between the orbit and
the constant-time surfaces.
\[
f_{[x]}(t) = f\left( [x] \cap \iota_t (L) \right) \ ,
\]
where $[x]$\ is the gauge orbit passing through some $x \in M_C$.

This time evolution takes place on the pre-symplectic manifold $M_C$\ which is
not the usual setting for describing a dynamical physical system. Moreover, it
is not a good starting point for standard quantization as functions on $M_C$\
do not have a well-defined Poisson bracket due to the degeneracy of
$i^*\Omega$. However, an evolving incisive section defines a family of
functions in $C^{\infty}(\mathcal{U})$, namely those that are constant along
the curves traced by points in $L$\ from one section to the next $x(t) =
\iota_t (y)$\ for a fixed $y \in L$. This family of functions can be
arbitrarily extended to the entirety of $M_C$:
\begin{defi} \label{FashClass} A subset ${\cal F}\subset C^{\infty}(M_C)$ is
  {\em fashionable} with respect to an evolving incisive section
  $(L,\Omega_L,\iota_t)$ if for all $t, t' \in (t_1, t_2)$\ the map
  $\iota_t^*\colon {\cal F}\to C^{\infty}(L)$ is a bijection and $\iota_t^* f
  = \iota_{t'}^*f$\ for all $f\in \mathcal{F}$.
\end{defi}
Given a choice of fashionables, each function on the symplectic manifold $L$\
corresponds to an evolving observable on $M_C$\ and conversely, the set of
evolving observables $\mathcal{F}$\ inherits a Poisson bracket from the
symplectic structure on $L$.

This notion of evolution has no known analog in the reduced phase space. In
practice it is usually implemented through {\em deparameterization}
\cite{GenHamDyn1,Blyth,RovelliTime,PartialCompleteObs,PartialCompleteObsII},
provided the constraint surface admits a factorization of the form $M_C\cong
\iota(L) \times {\mathbb R}\ni (\iota(x),Z)$ with a global coordinate
$Z\in{\mathbb R}$ such that $\{Z,C\}\not=0$. Then the map $\iota_t\colon L\to
M_C, x\mapsto (x,t)$\ defines a family of global incisive sections. Evolution
defined by this family of sections and ${\cal F}$, the $Z$-independent
functions on $M_C$, is called {\em global relational evolution} with respect
to $Z$.
\begin{ex}
 Let $M={\mathbb R}^{2(n+1)}\ni (Z,E,q_1,p_1,\ldots q_n,p_n)$ with
\[
 \Omega={\rm d}Z\wedge {\rm d}E+ \sum_{i=1}^n {\rm d}q_i\wedge {\rm d}p_i
\]
and a constraint $C=E+h(Z,q_i,p_i)$ linear in $E$. The constraint surface here
consists of points with coordinates $(Z, -h(q_i, p_i), q_i, p_i)$, so that
$(Z, q_i, p_i)$ serve as coordinates on $M_C$. The choice $L= {\mathbb
  R}^{2n}\ni (Q_1,P_1,\ldots Q_n,P_n)$ then leads to global incisive
deparameterized sections via $\iota_t\colon (Q_i, P_i) \mapsto (t, Q_i, P_i)
\in M_C$. Since for $C^{\infty}(M_C) \ni f=f(Z, q_i, p_i)$\ under this family
of embeddings $(\iota_t^* f)(Q_i, P_i) = f(t, Q_i, P_i)$, the corresponding
fashionables consist precisely of the functions that do not depend on
$Z$. Since the Hamiltonian vector field of $E$\ generates translations in $Z$\
and hence shifts from $\iota_t (L)$\ to $\iota_{t'}(L)$, the set of
fashionables correspond to the Poisson commutant $E'=\{f\in
C^{\infty}(M)\colon \{f,E\}=0\}$ of $E$ pulled back to $M_C$. Relational
evolution with respect to $Z$ is identical with Hamiltonian evolution in $L$
generated by $H_t(Q_i,P_i) = h(t, Q_i, P_i)$: The gauge flow $F_C(\epsilon)$
on $M$ maps a function $g\in C^{\infty}(M)$ to
$g_{\epsilon}=F_C(\epsilon)^*g$. In an infinitesimal version, $\delta g/\delta
\epsilon := \lim_{\epsilon\to0} (g_{\epsilon}-g)/\epsilon$ is given by
\[
  \frac{\delta g}{\delta \epsilon} = \{g,C\}= \frac{\partial g}{\partial
    Z}+ \{g,h\}\,.
\]
Specifically,
\[
 \frac{\delta q_i}{\delta \epsilon}= \frac{\partial h}{\partial p_i}\quad,\quad
 \frac{\delta p_i}{\delta \epsilon}= -\frac{\partial h}{\partial q_i}\quad,\quad
 \frac{\delta Z}{\delta \epsilon}=1\,.
\]
This pulls back to $L$ as
\[
 \frac{\delta Q_i}{\delta \epsilon}= \frac{\partial H_t}{\partial
   P_i}\quad,\quad 
 \frac{\delta P_i}{\delta \epsilon}= -\frac{\partial H_t}{\partial
   Q_i}\,. 
\]
For a function on $L$, we have
\begin{eqnarray*}
\frac{\delta f}{\delta \epsilon} &=& \lim_{\epsilon\to 0}
\frac{f(Q_i+\epsilon \delta Q_i/\delta\epsilon, P_i+\epsilon\delta
  P_i/\delta\epsilon)- f(Q_i,P_i)}{\epsilon}\\
&=&\sum_{i=1}^n\left( \frac{\partial f}{\partial 
   Q_i} \frac{\delta Q_i}{\delta \epsilon}+ \frac{\partial 
   f}{\partial P_i}\frac{\delta P_i}{\delta \epsilon}\right) =  \{f, H\}_L
\end{eqnarray*}
computed precisely according to Hamilton's equations on $L$.
\end{ex}

The quantization of a reduced phase space exists in the sense of deformation
quantization \cite{DefQuant1,DefQuant2} \`a la Fedosov or Kontsevich
\cite{Kontsevich}. On the other hand, dynamical symplectic reduction is
usually quantized only for deparameterized systems as in the immediately
preceding example, using a standard Hilbert-space quantization of $L$ on which
the reduced Hamiltonian $H_t(Q_i,P_i)$ is represented as an operator. In such
examples, quantum evolution exists and is unitary, but there are long-standing
problems when one tries to extend this notion to more complicated constrained
systems in which no global analog of $Z$ exists
\cite{KucharTime,Isham:Time}. For instance, given a constraint quadratic in
$E$ on the same phase space as in the example, $\{Z,C\}\propto E$ may become
zero along a gauge orbit such that $Z={\rm const}$ no longer defines a gauge
section. 

Heuristically, if there is no global analog of $Z$, evolution cannot be
represented by a family of unitary operators on a Hilbert space. Based on this
observation, we diagnose the main problem of standard approaches of
deparameterization as an over-reliance on Hilbert-space
representations. In order to solve this problem, we initiate a theory of {\em
  algebraic sections} as an algebraic quantization of classical gauge sections
in systems with a single constraint. By generalizing crucial steps of
deparameterization and keeping them strictly at the level of algebras of
observables, utilizing factor spaces of algebras instead of kernels of
constraint operators acting on a Hilbert space, we define a quantum version of
Proposition~\ref{PropClass} and derive
several necessary conditions that must be met by the constrained system in order for it to be dynamically reducible. While it remains difficult to find
sufficient conditions for such a result, the usefulness of our necessary
conditions is demonstrated by their restrictive nature in a specific example relevant to cosmology 
provided at the end of this paper. 

After setting the stage in Section~\ref{sec:AlgebraicApproach}, in
Section~\ref{sec:MainResults} we define algebraic qantization of dynamical
symplectic reduction and prove several properties of the resulting
quantum evolution on an algebra of observables. For deparameterizable systems,
which can be quantized by well-established means as representations on a fixed
Hilbert space, our algebraic results provide a more general treatment because
they apply to all possible choices of the Hilbert space. Moreover, our
construction applies to non-deparameterizable systems, even though the results
in that case are less specific than for deparameterizable systems. Several
results and examples in Section~\ref{sec:linearizing} will demonstrate the
non-trivial nature of our constructions.

\section{Preliminary constructions} \label{sec:AlgebraicApproach} 

To set the stage, we introduce in this section more details of the systems of
interest, standard procedures to quantize them as well as their limitations,
and basic algebraic ingredients to be used in our main results.

\subsection{Systems of interest}
\label{sec:Systems}

Quantum cosmology presents several versions of singly constrained systems,
resulting from the generally covariant theory of general relativity from which
their classical analogs are derived. Since space appears to be homogeneous and
isotropic over long distances, one can approximately describe the expansion of
the universe by a single time-dependent function, the scale factor $a(t)$,
subject to a single constraint. The latter is derived from the Friedmann
equation
\[
 \frac{1}{a^2}\left(\frac{{\rm d}a}{{\rm d}t}\right)^2=\frac{8\pi}{3} \rho(a)
\]
with the matter energy density $\rho(a)$ (written in units such that Newton's
constant and the speed of light equal one). In canonical form, a multiple of
the expansion rate $a^{-1}{\rm d}a/{\rm d}t$ is canonically conjugate to the
volume $V=a^3$,
\[
 p_V=-\frac{1}{4\pi a} \frac{{\rm d}a}{{\rm d}t}
\]
such that $\{V,p_V\}=1$.
The canonical Friedmann equation,  
\[
 -6\pi Vp_V^2+ E(V)=0
\]
can be written as a constraint equation which equates the matter energy
$E(V)=V\rho(V^{1/3})$ with a polynomial function of the volume and the
expansion rate. For non-relativistic matter (dominant at late times in the
universe), the energy density changes only by dilution in the expanding
universe, such that $\rho\propto 1/V$, or $E(V)=E$ constant.  The resulting
constraint
\begin{equation}\label{FriedmannC}
 C= -6\pi Vp_V^2+E\,,
\end{equation}
polynomial in $V$ and $p_V$ and linear in $E$, is a prototype of a large
set of models that have been studied to understand quantum evolution in
covariant systems. 

The dependence of $C$ on $V$ and $p_V$ varies according to the cosmological
model of interest. Moreover, there are additional anisotropy degrees of
freedom if one drops the assumption of spatial isotropy. Many models of this
type are known to have chaotic dynamics \cite{Billiards}, such that there is
usually no practical access to the reduced phase space $M/C$. Another set of
models is motivated by modified gravity, for instance the application of a
variety of quantization procedures, which may replace the Heisenberg algebra
generated by $V$ and $p_V$ by a different Lie algebra. Examples
of this type have been produced by models of loop quantum cosmology
\cite{LivRev,ROPP}, many of which can be formulated based on the Lie algebra
${\rm sl}(2,{\mathbb R})$ instead of the Heisenberg algebra generated by $V$
and $p_V$
\cite{BouncePert,CVHComplexifier,CVHPolymer,CVHProtected,NonBouncing}.

(If there are constraints in addition to $C$, which however do not contribute
to the dynamical flow, one can combine standard symplectic reduction with
dynamical symplectic reduction. For instance, in some cosmological models, the single variable $V$ could be
replaced by a pair $(V_1,V_2)$ with momenta $(p_{V_1},p_{V_2})$, subject to
rotational symmetry in the plane. A phase-space formulation would then
introduce a non-dynamical constraint $J=V_1p_{V_1}-V_2p_{V_2}=0$ given by
angular momentum in the plane. If the dynamical constraint $C$ extended to the
pair $(V_1,V_2)$ is rotationally symmetric, it has vanishing Poisson bracket
with the non-dynamical constraint, $\{C,J\}=0$, at least on the constraint
surface $J=0$ of $J$. In this example, and in many others of interest to gravity and cosmology, the Hamiltonian constraint is accompanied by finitely many additional constraints. In such cases standard
symplectic reduction of $J$ can easily be combined with dynamical symplectic
reduction of $C$, also at the quantum level.  A further generalization to
field theories would be required if one were to include perturbations around
isotropic cosmological models in order to describe inhomogeneous matter fields
on an expanding space-time background. Because time coordinates can be changed
locally in general relativity, there would then be not only an infinite number
of degrees of freedom, but also an infinite number of dynamical constraints,
one per point in space. Although such systems are certainly important for
modern cosmology, investigations of deparameterization in this context have
remained in their infancy. In particular, there seems to be no consensus so
far on the physical properties that should be required of deparameterization
in the presence of multiple dynamical constraints. In what follows, we will
therefore consider only the case of a single dynamical constraint.)

In order to introduce evolution in the constrained picture, it is common to
consider the constant $E$ as the momentum of a canonical variable $Z$ on which
the constraint does not depend. The constraint can then be quantized on the
kinematical Hilbert space ${\cal H}_{\rm kin}= L^2({\mathbb R},{\rm
  d}\lambda)\otimes {\cal H}_{\cal B}$ whose first factor is the Schr\"odinger
representation of the Heisenberg algebra generated by $E$ and $Z$, such that
$E=i\hbar \partial/\partial\lambda$, while ${\cal H}_{\cal B}$ is a unitary
representation of the algebra $\mathcal{B}$ generated by the original
canonical variables, such as the Heisenberg algebra generated by $V$ and
$p_V$.

If $C$ does not
depend on $Z$, as in (\ref{FriedmannC}), the standard procedure of
deparameterization, first suggested by Dirac \cite{GenHamDyn1} and applied to
quantum cosmology starting with \cite{DeWitt}, can be used to quantize the
dynamically constrained system: The constraint $C$ is represented as
an operator on ${\cal H}_{\rm kin}$ such that
\[
 C\psi= i\hbar\frac{\partial\psi(\lambda)}{\partial\lambda}-
 H\psi(\lambda) \ ,
\]
where $H$ is a self-adjoint representation of the $(V,p_V)$-dependent
contribution to $C$ on ${\cal H}_{\cal B}$. Zero eigenvectors of $C$ are
therefore given by
\[
 \psi(\lambda)=\exp(-i\lambda H/\hbar)\psi_0
\]
with arbitrary $\psi_0\in{\cal H}_{\cal B}$ as an ``initial state'' with
respect to evolution in $\lambda$. Because $U(\lambda)=\exp(-i\lambda
H/\hbar)$ is unitary, $\psi(\lambda)$ is not normalizable in ${\cal H}_{\rm
  kin}$, such that zero is in the continuous part of the spectrum of $C$.

In order to introduce a Hilbert-space structure on the solutions
$\psi(\lambda)$, we again make use of the unitarity of $U(\lambda)$ and define
a new inner product on the solution space by recycling the inner product
$(\cdot,\cdot)_{\cal B}$ on ${\cal H}_{\cal B}$: The physical inner product
\[
 (\psi(\lambda),\phi(\lambda))_{\rm phys}= (\psi_0,\phi_0)_{\cal B}
\]
turns the solution space into the physical Hilbert space ${\cal H}_{\rm phys}$
(which is not a subspace of ${\cal H}_{\rm kin}$). Unitarity of
$U(\lambda)$ implies that the inner product does not depend on
the choice of an initial $\lambda$-time:
\[
 (\psi(\lambda-\lambda_0),\phi(\lambda-\lambda_0))_{\rm phys}=
 (\psi_0,U(-\lambda_0)^{\dagger}U(-\lambda_0)\phi_0)_{\cal B}= (\psi_0,\phi_0)_{\cal
   B}\,.
\]
Moreover, any operator that commutes with $U(\lambda)$, called a Dirac
observable, has a unique representation on ${\cal H}_{\rm phys}$ However, it
is usually hard to compute Dirac observables or to show the existence of a
large-enough set, and they do not evolve because, by definition, they commute
with $U(\lambda)$.

In order to introduce an evolution picture on ${\cal H}_{\rm phys}$, one often
represents any operator $B\in{\cal B}$ by fixing an initial time, such that
the action of $B$ on $\psi_0\in{\cal H}_{\cal B}$ can be used. However, this
representation is not natural because choosing a different initial time
$\lambda_0$, such that $\psi_0$ is replaced by $U(-\lambda_0)\psi_0$, leads to
a different (but unitarily equivalent) representation unless $B$ commutes with
$U(-\lambda_0)$. With this construction, evolution is realized by the
$\lambda$-dependent expectation values
\[
 (\psi(\lambda),B\psi(\lambda))_{\rm phys}= (\psi_0,B\psi_0)_{\cal B}\,.
\]
The compution of Dirac observables can therefore be avoided if one accepts the
dependence of representations on the choice of an initial time. 

However, the applicability of these constructions is limited because they rely
on time-independent constraints which do not depend on the variable $Z$
canonically conjugate to $E$. Generic matter models in quantum cosmology and
other fields require such a dependence. For instance, in order to determine
evolution with respect to different choices of time coordinates, one would
have to interpret the gauge flows generated by $NC$ with some $N\in{\cal A}$,
rather than $C$ itself, and $N$ may well depend on $Z$ in cases of interest.
Moreover, relativistic matter systems imply an energy density quadratic in $E$
rather than linear, a prominent example given by a homogeneous scalar field
$Z$ (such as the inflaton often assumed in early-universe cosmology) with
energy density
\begin{equation} \label{ScalarDens}
 \rho_{\rm scalar}= \frac{1}{2}\frac{E^2}{V^2}+ W(Z)
\end{equation}
where the function $W(Z)$ is the scalar potential. The resulting constraint,
\[
 -12\pi V^2p_V^2+ E^2+2V^2W(Z)=0\,,
\]
is still polynomial in $V$ and $p_V$ but quadratic in $E$.  If $W(Z)=W$ is
constant, one can often ``take a square root'' and replace the constraint for
a scalar energy density with a constraint linear in $E$ by factorization,
\[
 \left(E-\sqrt{2}V \sqrt{6\pi p_V^2-W}\right) \left(E+\sqrt{2}V \sqrt{6\pi
     p_V^2-W}\right)=0\,,
\]
followed by selecting one of the two parentheses as a ``linearized''
constraint. To the new, linear constraint one can then apply
deparameterization as sketched above \cite{Blyth}. However, constant $W(Z)=W$
is not generic within the set of physically motivated models, and for
non-constant $W(Z)$ any factorization is non-trivial because $[W(Z),E]\not=0$.

Here, we propose an alternative way of reducing a quantum system with a
Hamiltonian constraint to a dynamical system, that solves most of these
problems and also reveals the non-trivial nature of introducing a well-defined
evolution picture. We completely avoid the construction of a physical Hilbert
space ${\cal H}_{\rm phys}$ distinct from the original, kinematical Hilbert
space ${\cal H}_{\rm kin}$. (Nevertheless, we will show that, if desired, a
physical Hilbert space can be derived from a subset of our ingredients using
the Gelfand-Naimark-Segal construction.) This feature brings our constructions
closer to a relativistic setting which seems violated in the construction
described above in which time is an operator represented only on ${\cal
  H}_{\rm kin}$, while all other observables are represented on ${\cal H}_{\rm
  phys}$. Our approach is based on an algebraic notion of quantum states.

\subsection{Algebraic states}

The set of observables of a quantum system is given by the $*$-invariant
elements of a complex, unital $*$-algebra $\mathcal{A}$. In this paper, we
assume that ${\cal A}$ is associative. (This assumption rules out some
physical systems, such as magnetic monopole densities \cite{Malcev,JackiwMon},
which however are usually considered exotic.) Our main examples will be
enveloping algebras of Lie algebras, which we assume to be represented on a
kinematical Hilbert space as unbounded operators. These algebras carry a
useful topology, introduced as the $\rho$-topology in \cite{UnboundedTop}. 

Physical states of the quantum system defined by ${\cal A}$ are normalized
positive linear functionals $\omega\colon \mathcal{A} \to {\mathbb C}$, such
that $\omega(\mathbf{1})=1$ and
\[
\omega \left(AA^*\right) \geq 0
\quad\mbox{for all } A\in \mathcal{A} \, . 
\]
According to Theorem~1 in \cite{UnboundedTop}, such functionals are continuous
in the $\rho$-topology of ${\cal A}$.  The condition that
$\omega\left(AA^*\right)$\ is real for all $A \in \mathcal{A}$
implies that a physical state is real ---
$\omega(A)=\overline{\omega(A^*)}$. In addition, the stronger inequality
condition leads to the Cauchy--Schwarz inequality
\[
 |\omega(AB^*)|^2 \leq |\omega(AA^*)| |\omega(BB^*)| \quad\mbox{ for all } A,
 B \in  \mathcal{A} \,;
\]
see for instance \cite{LocalQuant}.

As we will see, intermediate stages of quantum symplectic reduction require a
weaker notion of states which are not completely positive. We begin with
\begin{defi}\label{def:AlgebraicStates}
  The set of {\em kinematical states} $\Gamma$\ on a unital $*$-algebra
  $\mathcal{A}$ is the set of continuous normalized linear functionals
  $\omega\colon \mathcal{A}\to {\mathbb C}$, such that $\omega(\mathbf{1})=1$.
\end{defi}
Given the normalization condition, $\Gamma$ is not a vector space, but it is
closed with respect to normalized sums: for any integer $N\geq 1$, states
$\omega_1,\ldots,\omega_N\in\Gamma$ and complex numbers $a_1,\ldots,a_N$,
$\sum_{j=1}^N a_j\omega_j \in \Gamma$ if $\sum_{j=1}^Na_j = 1$.

\begin{defi} \label{def:Dflow} A {\em dynamical flow} on $\mathcal{A}$\ is a
  one-parameter family of derivations $\vec{D}_t\colon (a, b) \times
  \mathcal{A} \rightarrow \mathcal{A}$, where $(a, b) \subset \mathbb{R}$,
  which is compatible with the $*$-structure on $\mathcal{A}$ --- $(\vec{D}_t
  A)^* = \vec{D}_t A^*$ for all $A \in \mathcal{A}$ --- and such that $\omega(
  \vec{D}_t A)$ is continuously differentiable with respect to $t$ for all
  $\omega \in \Gamma$.

  Given a dynamical flow $\vec{D}_t$ on ${\cal A}$, the {\em time evolution}
  of a kinematical state $\omega\in\Gamma$ is a map $(a,b)\times{\cal
    A}\to{\mathbb C}, (t,A)\mapsto\omega_t(A)$ such that $\omega_t$ is a
  kinematical state and
\begin{equation} \label{eq:Ddynamics}
\frac{{\rm d}}{{\rm d}t} \omega_t(A) = \omega_t \left( \vec{D}_t A \right)
\end{equation}
for all $t\in(a,b)$, with initial conditions $\omega_{t_0}=\omega$\ for some
$t_0\in(a,b)$.
\end{defi}

In order to make sure that a state has a unique time evolution (or a unique
gauge flow in what follows), we will assume that, for all algebras we
consider, a differential equation of the form (\ref{eq:Ddynamics}) has a
unique solution with the specified initial condition. Standard results do not
necessarily apply because our differential equations, though linear, are, in
general, formulated on an infinite-dimensional space and may have
time-dependent coefficients. (Although we will not pursue a formal proof of
existence and uniqueness of solutions, we note that time evolution in systems
of interest in physics is usually obtained as a unique Dyson series on a
Hilbert space; see for instance \cite{DysonSeries}.)

\begin{lemma}\label{EvolutionLemmaKin}
  If $\omega\in\Gamma$ is a kinematical state, its time evolution with respect
  to $\vec{D}_t$, $t\in (a,b)$, returns a kinematical state for any $t\in
  (a,b)$.
\end{lemma}
\begin{proof}
 By definition, a derivation satisfies
\begin{equation} \label{eq:Derivation}
\vec{D}_t \left( AB \right)  = \vec{D}_t \left(A\right) B +A\vec{D}_t
\left(B\right) 
\end{equation}
for all $A,B\in {\cal A}$. Choosing $B=\mathbf{1}$, we have
$\vec{D}_t(A)=\vec{D}_t(A)+ A\vec{D}_t(\mathbf{1})$ for all $A\in{\cal A}$. It
follows that $\vec{D}_t(\mathbf{1})=0$, whence ${\rm
  d}\omega_t(\mathbf{1})/{\rm d}t=0$ for all $t$. Therefore,
$\omega_t(\mathbf{1})=1$ for all $t$.
\end{proof}

\begin{lemma} \label{EvolutionLemma}
 If $\omega\in\Gamma$ is positive, its time evolution is positive.
\end{lemma}
\begin{proof}
To prove that $\omega_t(AA^*) \geq 0$\ continues to hold along the flow, it is sufficient to show that (i) $\omega_t(AA^*)$\ is real for all $t$\ and (ii) ${\rm d} \omega_t(AA^*)/{\rm d}t$ is non-negative whenever $\omega_t(AA^*) = 0$.
 
To prove (i), for each $A\in \mathcal{A}$\ define a function of $t$\ via
$f_A(t) = \omega_t(AA^*) - \overline{\omega_t(AA^*)}$\ on $\Gamma$, so that
$\omega_t(AA^*)$\ is real iff $f_A(t)=0$. Suppose all of the functions
$f_A(t')=0$\ for some $t'\in (a, b)$, then $\omega_{t'}(AA^*)$\ is real for
all $A\in \mathcal{A}$, which implies
$\omega_{t'}(A)=\overline{\omega_{t'}(A^*)}$, and we get
\begin{eqnarray*}
\left. \frac{{\rm d}}{{\rm d} t} \omega_t (AA^*) \right|_{t=t'}&=& \left. \omega_t\left( \vec{D}_t (AA^*) \right) \right|_{t=t'}= \omega_{t'}\left( \left(\vec{D}_{t'} A\right) A^* \right) + \omega_{t'}\left( A \left(\vec{D}_{t'} A^*\right) \right) \\
&=& \omega_{t'}\left( \left(\vec{D}_{t'} A\right) A^* \right)  + \omega_{t'} \left(
  \left[ \left(\vec{D}_{t'} A\right) A^*\right]^* \right)\\
&=& \omega_{t'}\left( \left(\vec{D}_{t'} A\right) A^* \right)+
\overline{\omega_{t'}\left( \left(\vec{D}_{t'} A\right) A^* \right)}\\
&=& 2 {\rm Re} \left[ \omega_{t'}\left( \left(\vec{D}_{t'} A\right) A^* \right)\right] = 2 {\rm Re} \left[ \overline{\omega_{t'}\left( \left(\vec{D}_{t'} A\right) A^* \right)}
\right] = \left. \frac{{\rm d}}{{\rm d} t} \overline{\omega_t (AA^*)} \right|_{t=t'} \ ,
\end{eqnarray*}
which means that ${\rm d}f_A(t)/{\rm d}t = 0$\ at $t=t'$. Since
$\omega_{t_0}=\omega$\ is positive, we have the initial conditions
$f_A(t_0)=0$\ for all $A \in \mathcal{A}$. We see that $\{f_A(t)=0, \forall t
\in (a, b)\}_{A\in \mathcal{A}}$\ satisfies the first-order ordinary
differential equation system induced by the dynamical flow and matches the
given set of initial conditions. As previously discussed, here we assume such
solutions to the dynamical flow to be unique. Therefore $\omega_t(AA^*)$\ is
real for all $t\in(a, b)$.

To prove (ii) we use the above result and assume that the inequality holds at $t=t'$.
\begin{eqnarray*}
\left| \left. \frac{{\rm d}}{{\rm d} t} \omega_t(AA^*) \right|_{t=t'} \right|^2 &=& 4 \left| {\rm
    Re} \left[ \omega_{t'}\left( \left(\vec{D}_{t'} A\right) A^* \right) \right]
\right|^2 \\
&\leq& 4 \left| \omega_{t'}\left( \left(\vec{D}_{t'} A\right) A^* \right) \right|^2
\\
&\leq& 4 \, \omega_{t'} (A^*A)\,\, \omega_{t'} \left( \left(\vec{D}_{t'}
    A\right)\left(\vec{D}_{t'} A\right)^*\right) \ . 
\end{eqnarray*}
Since $\omega_{t'}(A^*A) \in \mathbb{R}$, $\omega_{t'}(A^*A) =
\omega_{t'}(AA^*)$, and the expression on the right is zero if
$\omega_{t'}(AA^*)=0$.
\end{proof}

\subsection{Constrained quantization}\label{sec:QConstraint}

Our results in Sections~\ref{sec:MainResults} and \ref{sec:linearizing} apply
to a specific type of constrained systems relevant for quantizations of
dynamical symplectic reduction.
\begin{defi}
A singly constrained quantum system is a complex, unital $*$-algebra
$\mathcal{A}$ together with a constraint $C\in\mathcal{A}$ such that 
\begin{enumerate}
\item $C^*=C$,
\item $C$ does not have a left-inverse in ${\cal A}$, and
\item $C$ is not a divisor of zero.
\end{enumerate}
\end{defi}
\begin{rem}
  If $C$ is a right divisor of zero, then by property 1 it is also a left
  divisor of zero and vice versa, in which case there is an $X\in{\cal A}$
  such that $CX=0$. Using a representation of ${\cal A}$ on the kinematical
  Hilbert space, any vector in the image of $X$ is then a zero eigenvector of
  $C$ in the discrete spectrum, and one can simply solve the constraint $C=0$
  by restriction to the zero eigenspace. Here we are primarily concerned with
  the more complicated case of zero in the continuous part of the spectrum of
  $C$, which is also most relevant for dynamical symplectic reduction of
  examples discussed in Section~\ref{sec:Systems}.
\end{rem}

\begin{defi}\label{def:A-obs}
  The {\em algebra of Dirac observables} of a singly constrained quantum
  system $(\mathcal{A},C)$ is the commutant of $C$ in $\mathcal{A}$:
  \[
 \mathcal{A}_{\rm obs}=C'=\{A\in\mathcal{A} \, :\, [A, C] = 0 \}\,.
  \]

  The space of {\em physical states} of a singly constrained quantum system is
  given by the space $\Gamma({\cal A}_{\rm obs})$ of normalizaed positive
  linear functionals on ${\cal A}_{\rm obs}$.
\end{defi}
\begin{lemma}\label{lem:Comutant}
 $\mathcal{A}_{\rm obs}$ is a unital $*$-subalgebra of $\mathcal{A}$.
\end{lemma}
\begin{proof}
  Defined as the commutant of $C$, ${\cal A}_{\rm obs}$ is a subalgebra. Since
  $[\mathbf{1},C]=0$ and $[A^*,C]=-[A,C^*]^*=-[A,C]^*=0$ if $A\in{\cal A}_{\rm
    obs}$, using $C^*=C$, it is a unital $*$-subalgebra.
\end{proof}

\begin{defi}
  A kinematical state $\omega \in \Gamma$\ is a {\em solution of the
    constraint $C$} if $\omega(AC) = 0$ for all $A \in \mathcal{A}$.  The {\em
    constraint surface} $\Gamma_C\subset\Gamma$ is the subset of all
  solutions of $C$, closed with respect to normalized sums.
\end{defi}

\begin{rem}
  Since we have assumed that $C$ is without left-inverse in ${\cal A}$, ${\cal
    A}C\subset {\cal A}$ is a strict subalgebra without unit. The condition
  $\omega(AC)=0$ is therefore consistent with normalization of kinematical
  states. If zero is in the continuous spectrum, no normalized positive states
  exist in $\Gamma_C$. In such a case it is common to drop the normalization
  condition and work with distributions instead of state vectors in a Hilbert
  space. Here, we instead retain the normalizability condition and relax
  positivity. Our results (such as Corollary~\ref{cor:Positive} in
  Section~\ref{sec:Positivity}) imply the existence of kinematical states in
  $\Gamma_C$ for singly constrained systems.
\end{rem}

The constraint $C$ in a singly constrained quantum system
induces a gauge flow:
\begin{defi} \label{def:Cequiv} Two kinematical states $\psi, \omega \in
  \Gamma$ are {\em $C$-equivalent}, $\omega \thicksim_C \psi$, if there exist
  a positive integer $M$ together with $A_1, A_2, \ldots, A_M \in \mathcal{A}$
  and $\lambda_1, \lambda_2, \ldots \lambda_M \in \mathbb{R}$, such that
\[
\psi = S_{A_1C} (\lambda_1) S_{A_2C} (\lambda_2) \ldots S_{A_MC} (\lambda_M)
\omega 
\]
where for $A\in\mathcal{A}$ and $\lambda\in{\mathbb R}$, the {\em flow}
$S_A(\lambda)\colon \Gamma\to\Gamma$ is defined by $S_A(0)={\rm id}$ and
\[
i\hbar \frac{{\rm d}}{{\rm d} \lambda} \left( S_A(\lambda) \omega(B) \right)
=  S_A(\lambda) \omega([B, A]) \ . 
\]
\end{defi}

Since $B\mapsto [B,A]$ is a derivation, $S_A(\lambda)$ is well-defined by
Lemma~\ref{EvolutionLemmaKin}. By analogy with classical reduction, we refer
to flows generated by elements of $\mathcal{A}C$\ as gauge. In the classical
case, the constraint $C$ and the function $fC$, where $f$ is any phase-space
function not equal to zero on the constraint surface, generate the same gauge
flow on the constraint surface: For $fC$, the flow equations
\[
 \frac{{\rm d}}{{\rm d}\lambda'}= \{\cdot,fC\}\approx f\{\cdot,C\}
\]
can be mapped to the flow equations generated by $C$ by rescaling the gauge
parameter $\lambda'$ such that ${\rm d}\lambda=f{\rm d}\lambda'$. (For this
reparameterization of the flow, the phase-space function $f$ is considered to
be a function of $\lambda'$ on the space of solutions of the flow
equations. The sign ``$\approx$'' in the preceding equations refers to
equality on the constraint surface, as usual in the theory of constrained
systems following Dirac.)  In singly constrained quantum systems, by contrast,
the gauge flows generated by $C$ and $AC$ with some $A\in{\cal A}$ are
inequivalent: For $\omega\in\Gamma_C$, we have
\[
 \omega([B,AC])= \omega(A[B,C])\not= \omega(A)\omega([B,C])
\]
in general. 

\begin{rem}
  If zero is in the discrete spectrum of $C$ represented on the kinematical
  Hilbert space, the behavior of flows is rather different. Any quantum flow
  generated by $AC$ is then equivalent to the flow generated by $C$ when
  restricted to constrained \emph{pure} states, given by vectors $\psi$ in the
  Hilbert space such that $C\psi=0$: on these states,
  $S_{AC}(\lambda)\psi=\exp(-i\lambda AC/\hbar)\psi$ such that
\[
 S_{AC}(\lambda) \psi= \sum_{n=0}^{\infty} \frac{1}{n!}
 \left(\frac{-i\lambda}{\hbar}\right)^n (AC)^n\psi=0
\]
for any $A$. The full gauge flow may nevertheless be non-trivial on general
algebraic states of Definition~\ref{def:AlgebraicStates}.
\end{rem}
 
\begin{lemma} \label{lem:GammaCinv} The constraint surface $\Gamma_C$ is
  preserved by the flow induced by any algebra element $AC$.
\end{lemma}
\begin{proof} For any fixed $A\in \mathcal{A}$\ and $\omega \in \Gamma_C$,
  following the same argument as in Lemma~\ref{EvolutionLemma}, define
  functions $f_B(\lambda) = S_{AC}(\lambda) \omega(BC)$\ on $\Gamma$, for
  $B\in\mathcal{A}$. Suppose all $f_B(\lambda') = 0$\ for some $\lambda'$,
  then
\begin{eqnarray*}
\left.  i\hbar \frac{{\rm d} f_B}{{\rm d}\lambda} \right|_{\lambda = \lambda'}
&=& \left. i\hbar \frac{{\rm d}}{{\rm d} \lambda} \left( S_{AC}(\lambda) \omega(BC) \right) \right|_{\lambda = \lambda'}
\\
&=&  \left. S_{AC}(\lambda) \omega([BC, AC]) \right|_{\lambda = \lambda'}
\\
&=& S_{AC}(\lambda')\omega\left( ([BC,A] + A[B, C] ) C \right)
= f_{([BC,A] + A[B, C])}(\lambda') = 0
\end{eqnarray*}
for all $B\in\mathcal{A}$.  Moreover, we have initial conditions $f_B(0) =
\omega(BC) = 0$\ for all $B$. It follows that $\{f_B(\lambda)=0, \forall
\lambda \}_{B\in \mathcal{A}}$\ is the solution to the flow induced by any
algebra element of the form $AC$\ that satisfies our initial
conditions. Therefore, $S_{AC}(\lambda) \omega(BC)=0$ for all $\lambda$, and
$S_{AC}(\lambda) \omega \in \Gamma_C$.
\end{proof}

Any two $C$-equivalent states on $\Gamma_C$\ are indistinguishble by their
evaluation in Dirac observables:

\begin{lemma} \label{lemma:Cequiv} For any $\omega, \psi \in \Gamma_C$, if
  $\omega\thicksim_C \psi$, then $\omega(O)=\psi(O)$ for any
  $O\in\mathcal{A}_{\rm obs}$.
\end{lemma}

\begin{proof} The two states $\omega$ and $\psi$ are related by a succession
  of gauge flows $S_{AC}(\lambda)$. By Lemma~\ref{lem:GammaCinv}, each of
  these flows preservers $\Gamma_C$. Therefore, for any $A\in\mathcal{A}$ and
  $B\in{\cal A}_{\rm obs}$,
\begin{eqnarray*}
 i\hbar \frac{{\rm d}}{{\rm d} \lambda} \left( S_{AC}(\lambda) \omega(B)
  \right)  &=&  S_{AC}(\lambda) \omega([B, AC])
 \\  &=&  S_{AC}(\lambda)\omega\left( A[B, C] + [B, A] C
\right)\\ &=&  S_{AC}(\lambda)\omega\left([B, A] C \right)= 0  \ ,
\end{eqnarray*}
since $ S_{AC}(\lambda)\omega\in\Gamma_C$. Therefore, $S_{AC}(\lambda)
\omega(B)$ is constant along any gauge flow $S_{AC}(\lambda)$.
\end{proof}

Equivalence classes $[\omega]_C\in\Gamma_C/\thicksim_C$ therefore define
states on ${\cal A}_{\rm obs}$.
\begin{cor}\label{def:PhysStates}
  The space of {\em physical states} $\Gamma_{\rm
    phys}$ is the convex subset of $\Gamma_C/\thicksim_C$ containing all
  $[\omega]_C$ with $\omega$ positive on $\mathcal{A}_{\rm obs}$.
\end{cor}
The computation of a complete ${\cal A}_{\rm obs}$ is usually very complicated
in interesting models. Moreover, a sufficiently large ${\cal A}_{\rm obs}$
containing observables that can describe all measurements of interest may not
exist, in particular in chaotic systems \cite{DiracChaos,DiracChaos2}. The
result of Corollary~\ref{def:PhysStates} partially avoids a direct reference
to ${\cal A}_{\rm obs}$ by formulating the space of physical states through an
equivalence relation on the constrained states. However, one still needs
access to ${\cal A}_{\rm obs}$ in order to implement the positivity
condition. This reference to ${\cal A}_{\rm obs}$ cannot be avoided because it
is generally impossible to extend positivity to all of ${\cal A}$ for any
$\omega\in\Gamma_C$: for any $A=A^* \in \mathcal{A}$\ and a positive $\omega
\in \Gamma$, we have $\omega(AC+CA) \in \mathbb{R}$\ and $\omega([A, C]) \in
i\mathbb{R}$. However for $\omega \in \Gamma_C$\ we must have 
\[
 \omega(AC+CA) = \omega(2AC-[A, C]) = -\omega ([A,C])\,,
\]
which, given the reality conditions, can be satisfied only if $\omega(AC+CA) =
\omega ([A,C]) = 0$. If $C$\ possesses a canonical conjugate
$Z=Z^*\in\mathcal{A}$\ such that $[Z, C] = i\hbar\mathbf{1}$ (as we will
shortly demand for a ``clock''), then for any normalized state $\omega([Z,C])
= i\hbar \neq 0$, so that no solution of the constraint is positive on all of
$\mathcal{A}$. Moreover, as in the classical case, there is no evolution for
physical states in a dynamically constrained system, since the adjoint action
of the constraint has been factored out.

In Section~\ref{sec:MainResults} we will solve both problems --- formulating
positivity conditions without reference to ${\cal A}_{\rm obs}$ and obtaining
a consistent evolution picture --- by introducing a new notion of gauge
sections. Our approach is motivated by the algebraic analogue of
deparameterization discussed in Section~\ref{sec:Systems}. In an unconstrained
quantum system the dynamical flow is usually driven by some self-adjoint
Hamiltonian $H=H^*\in{\cal B}$\ via
\begin{equation}\label{eq:Hflow}
 \frac{{\rm d}\omega_t(B)}{{\rm d} t} = \frac{1}{i\hbar} \omega_t \left( [B,H]
 \right) + \omega_t \left( \frac{\rm d}{{\rm d} t} B \right)\,, 
\end{equation}
where the derivation $\frac{1}{i\hbar} [\cdot,H] + \frac{\rm d}{{\rm d}
  t}(\cdot)$\ clearly satisfies Definition~\ref{def:Dflow}. The above
evolution of states can be formulated as a pure adjoint action of a constraint
by extending the kinematical algebra ${\cal B}$ by two new generators,
``time'' $Z=Z^*$ and ``energy'' $E=E^*$, such that $[Z,E]=i\hbar \mathbf{1}$
and $[Z,B]=0=[E,B]$ for all $B$ in $\mathcal{B}$. On this extended algebra
$\mathcal{A}$, the constraint $C:=E+H\in{\cal A}$ generates a gauge flow
\begin{equation} \label{OCflow} \frac{\rm d}{{\rm d} \lambda} \omega_{\lambda} (A)
  =\frac{1}{i\hbar} \omega_{\lambda} \left([A,C] \right) = \frac{1}{i\hbar}
  \omega_{\lambda}\left( [A,H] + [A,E] \right) = \frac{\omega_{\lambda} \left( [A,H]
  \right)}{i\hbar} + \omega_{\lambda} \left( \frac{\rm d}{{\rm d} Z} A \right)
\end{equation}
resembling the original dynamical flow on $\mathcal{B}$ for any $A$ polynomial
in $Z$. Explicit time dependence of $B(t)$ in the Hamiltonian case corresponds
to $Z$-dependence of $A\in{\cal A}$ in the constrained case. This process is
called parameterization of the dynamical flow: physical time $t$\ has been
replaced by an arbitrary flow parameter $\lambda$.

One can follow this process in reverse: starting with the degrees of freedom
described by the extended algebra $\mathcal{A}$\ and given a Hamiltonian
constraint one could attempt to \emph{deparameterize} the system, reducing it to a
an unconstrained dynamical system on a subalgebra $\mathcal{B}\subset \mathcal{A}$. In the
above parameterized theory, the gauge flow (\ref{OCflow}) is equivalent to the
dynamical flow (\ref{eq:Hflow}) if $Z\in{\cal A}$ can be ``demoted to a real
number'' $t$.  The deparameterization process of passing from $\mathcal{A}$
back to the smaller algebra $\mathcal{B}$ can therefore be interpreted as
finding the states on $\mathcal{A}$ for which ``the value of $Z$ is
fixed'' to equal $t$. For a general Hamiltonian constraint, this process
requires finding a suitable clock $Z=Z^*$: its values will keep track of time,
and an associated algebra (the fashionables of Definition~\ref{FashClass})
will play the role of the smaller ``unextended'' algebra of evolving degrees
of freedom.

\subsection{Quantum clocks}\label{sec:QClocks}

Any kinematical observable $Z=Z^* \in \mathcal{A}$\ can formally serve as a
quantum reference system. In Section~\ref{sec:MainResults} we will use such
reference systems to track translations in time, where $Z$\ will serve as a
``clock.'' (Reference systems can be used to track spatial translations as
well \cite{QuantumRef1,QuantumRef2,QuantumRef3,QuantumRef4}.) Following
Lemma~\ref{lem:Comutant} the commutant $Z'=\left\{ A \in \mathcal{A} : [A, Z]
  = 0 \right\}$\ is a unital $*$-subalgebra of $\mathcal{A}$. Furthermore, it
is straightforward to verify that for any $t\in\mathbb{R}$\ the set
$(Z-t\mathbf{1})Z'$\ is a two-sided $*$-ideal of $Z'$. It follows that the
quotient $Z'/(Z-t\mathbf{1})Z'$\ is unital and inherits a $*$-structure under
the canonical projection $\pi_t\colon Z' \rightarrow Z'/(Z-t\mathbf{1})Z'$. In
order to be useful for keeping track of time, such a quantum reference system
needs some additional structure.

\begin{defi} \label{def:QClock} A \emph{quantum clock} $(Z, \mathcal{F})$\ is a reference observable $Z=Z^* \in \mathcal{A}$\ together with a compatible \emph{fashionable algebra}: a unital $*$-subalgebra ${\cal F} \subset Z'$, such that for all $t \in  \mathbb{R}$, we have $\mathcal{F} \cap \ker \pi_t = \{ 0\}$ and $\pi_t(\mathcal{F}) = Z'/(Z-t\mathbf{1})Z'$.
\end{defi}
\begin{rem} Keeping with the commonsense physics usage, the term ``clock'' will also be used to refer to the reference observable $Z$\ itself, where it will be assumed that $Z$\ possesses a compatible fashionable algebra.\end{rem}

The two conditions on $\mathcal{F}$\ guarantee that $\pi_t$ restricted to $\mathcal{F}$ is a
$*$-algebra isomorphism. The algebra of fashionables is therefore a
realization of a family of quotient algebras $Z'/(Z-t\mathbf{1})Z'$ as a
single subalgebra of $Z'$ (and hence of $\mathcal{A}$).  We denote the
$*$-isomorphism $\nu_t\colon Z'/(Z-t\mathbf{1})Z' \rightarrow \mathcal{F}$,
where for any $X \in Z'/(Z-t\mathbf{1})Z'$,
\begin{equation}\label{eq:ZtoFmap}
\nu_t(X) := \pi_t^{-1} (X) \cap \mathcal{F}
\end{equation}
yields a single element of $\mathcal{F}$. This isomorphism inverts $\pi_t$
when the latter is restricted to $\mathcal{F}$, so that $\pi_t \circ \nu_t =
{\rm id}$. As a direct consequence, we  note
\begin{lemma} \label{CosetLemma}
 For every $t\in{\mathbb R}$,  $\mathcal{F}+(Z-t\mathbf{1})Z' = Z'$. 
\end{lemma}
Additionally, for each value of $t$ we have a projection from $Z'$ to its
subalgebra $\mathcal{F}$ via the composition of $*$-homomorphisms $\nu_t \circ
\pi_t$.

Each quotient algebra $ Z'/(Z-t\mathbf{1})Z'$\ possesses a natural relational interpretation:
\begin{defi} For any kinematical observable $A=A^* \in \mathcal{A}$ that
  commutes with a clock $Z$, the {\em relational observable} for $A$ when the
  value of $Z$ is $t \in \mathbb{R}$ is the image of $A$ under the canonical
  projection, $A_{Z=t} := \pi_t(A)$.
\end{defi}
It is straightforward to check that $A_{Z=t} = A_{Z=t}^*$\ under the inherited
$*$-operation.

\begin{rem} Time read by a reference clock behaves differently from the time
  of ordinary quantum mechanics where it is an external parameter. In the
  latter case any observable can be evaluated at a fixed time, in the former
  case this can only be done for observables that commute with the
  clock.\end{rem}

We will use the notation $\bar{\omega} \in \Gamma_{Z'}$\ to denote normalized
states on $Z'$.  For each $t$, canonical projection $\pi_t$ defines a
``surface of constant time'' within $\Gamma_{Z'}$, which will be used in
Section~\ref{sec:Transversality} to fix the gauge flow of a quantum
Hamiltonian constraint.

\begin{defi}  \label{GaugeFixedStates} For any $t \in \mathbb{R}$\ the \emph{surface of constant time $t$\ of a clock $Z$}\ is
\begin{equation*}
 \Gamma_{Z'}|_{\pi_t} = \{ \bar{\omega} \in \Gamma_{Z'} : \bar{\omega}(B)=0
 \mbox{ for all }B\in\ker \pi_t\}\,. 
\end{equation*}
\end{defi}
Each state on this surface is a pull-back $\Gamma_{Z'}|_{\pi_t} \ni
\bar{\omega}=\tilde{\omega}\circ\pi_t$\ of a state $\tilde{\omega}$\ on the
quotient algebra $Z'/(Z-t\mathbf{1})Z'$. The states belonging to such a
surface that are positive on $Z'$\ can be interpreted as assigning a value to
each $A=A^* \in Z'$\ ``when clock $Z$\ shows time $t$''. Not every positive
state $\bar{\omega}$\ on the algebra $Z'$ has a relational interpretation
since this would require $\bar{\omega}((Z-t\mathbf{1})A)=0$\ for any $A \in
Z'$\ for some $t \in \mathbb{R}$. A sequence of relational states 
provides a notion of evolution over time in clock $Z$.

\begin{defi}\label{def:RelationalEvol} A one-parameter family of states
  $\bar{\omega}_t \in \Gamma_{Z'}$\ for $t\in (a,b) \subset \mathbb{R}$\ is a
  \emph{time evolution of $\bar{\omega} \in \Gamma_{Z'}$ relative to $Z$}\ if
  $\bar{\omega} =\bar{\omega}_{t_0}$\ for some $t_0\in(a, b)$, and if for each
  $t\in (a,b)$\ $\bar{\omega}_t$\ is positive on $Z'$ and $\bar{\omega}_t \in
  \Gamma_{Z'}|_{\pi_t}$.
\end{defi}

While any positive state on a quotient $Z'/(Z-t\mathbf{1})Z'$\ does have a
relational interpretation, these quotients give distinct algebras for
different values of $t$. So a time evolution in $Z$\ corresponds to a sequence
of states on different disconnected algebras. The fashionable algebra is
needed precisely to ``sew together'' constant time degrees of freedom at
different values of the clock and to define clock-value-dependent states that
can be freely specified over a fixed algebra $\mathcal{F}$. We define a
one-parameter family of invertible linear maps $\psi_t:
\Gamma_{Z'}|_{\pi_t}\rightarrow \Gamma_{\mathcal{F}}$\ with the forward map
given simply by restriction, for $\bar{\omega}\in \Gamma_{Z'}|_{\pi_t}$
\begin{equation} \label{psit}
\psi_t(\bar{\omega})(F) = \bar{\omega}(F) \ , \ \ {\rm for\ all\ } F \in \mathcal{F}\ .
\end{equation}
We define its inverse using~(\ref{eq:ZtoFmap}), for $\tilde{\omega} \in \Gamma_{\mathcal{F}}$
\begin{equation}
\psi_t^{-1}(\tilde{\omega})(A) = \tilde{\omega}\left( \nu_t(\pi_t(A)) \right) \ , \ \ {\rm for\ all\ } A \in Z'\ .
\end{equation}
Clearly, because it involves $\pi_t$, the resultant state
$\psi_t^{-1}(\tilde{\omega}) \in \Gamma_{Z'}|_{\pi_t}$. It is also not
difficult to verify that $\psi$\ and $\psi^{-1}$\ invert each other on their
domains of definition. The fashionable algebra gives us the structure
necessary to define time translation of a state from $Z=t_1$ to $Z=t_2$:
$\bar{\omega}_1 \in \left. \Gamma_{Z'} \right|_{\pi_{t_1}}$ and
$\bar{\omega}_2 \in \left. \Gamma_{Z'} \right|_{\pi_{t_2}}$ represent the same
unevolved state at two different times $t_1$ and $t_2$ if $\bar{\omega}_1 (F)
= \bar{\omega}_2 (F)$ for all $F\in \mathcal{F}$, that is, if
$\bar{\omega}_1=\psi_{t_1}^{-1} (\tilde{\omega})$\ and
$\bar{\omega}_2=\psi_{t_2}^{-1} (\tilde{\omega})$\ for some $\tilde{\omega}
\in \Gamma_{\mathcal{F}}$.

A time evolution maps to a one-parameter family of positive states on the
fashionable algebra via $\tilde{\omega}_t
=\psi_t(\bar{\omega}_t)$. Conversely, any one-parameter family of positive
states $\tilde{\omega}_t \in \Gamma_{\mathcal F}$\ maps to a time evolution
with respect to $Z$ via $\bar{\omega}_t = \psi_{t}^{-1}(\tilde{\omega}_t)$.
The maps $\psi_t$\ and $\psi_{t}^{-1}$\ preserve positivity of states, because
the underlying algebra maps that they utilize, $\imath \colon \mathcal{F}
\hookrightarrow Z'$, $\pi_t$, and $\nu_t$\ are all $*$-homomorphisms. This can
be stated somewhat more broadly.
\begin{lemma}\label{lem:RelationalPositivity} If $\bar{\omega} \in
  \Gamma_{Z'}|_{\pi_t}$\ for some $t\in \mathbb{R}$\ and $\bar{\omega}$\ is
  positive on $\mathcal{F}\subset Z'$, then $\bar{\omega}$\ is positive on the
  whole of $Z'$.
\end{lemma}
\begin{proof}
  Since $\bar{\omega} \in \Gamma_{Z'}|_{\pi_t}$, we have
  $\bar{\omega}((Z-t\mathbf{1})A) = 0$\ for all $A\in Z'$. According to
  Lemma~\ref{CosetLemma}, for any $B \in Z'$ there are $F\in {\cal F}\subset
  Z'$ and $B_0\in Z'$ such that $B=F+(Z-t\mathbf{1})B_0$. Thus
\begin{eqnarray*}
\bar{\omega} \left( BB^* \right) &=& \bar{\omega} \Bigl( FF^* +  F(Z-t\mathbf{1})B_0^* +
(Z-t\mathbf{1}) B_0 F^*+ (Z-t\mathbf{1}) B_0 (Z-t\mathbf{1})B_0^* \Bigr) 
\\
&=& \bar{\omega} \left( FF^* \right) + \bar{\omega}\Bigl( (Z - t \mathbf{1}) \left( FB_0^*+ B_0F^*  +  (Z-t\mathbf{1})B_0 B_0^* \right) \Bigr)
\\
&=& \bar{\omega} \left( FF^* \right) \geq 0 \ .
\end{eqnarray*}
\end{proof}

In what follows we use a clock as the basis for a new method of
deparameterization by (i) using the commutant $Z'\subset {\cal A}$ as a
go-between of the kinematical algebra ${\cal A}$ and the algebra of
observables, ${\cal A}_{\rm obs}$, (ii) fixing gauge degrees of freedom using
constant time surfaces $\Gamma_{Z'}|_{\pi_t}$, and (iii) demoting the clock
observable $Z$\ to a real-valued time parameter $t$.  The constant-time
surfaces are analogous to fixing an initial $\lambda$-time in standard
deparameterization discussed in Section~\ref{sec:Systems}, but they are
introduced in a subalgebra of ${\cal A}$ and do not require the transition to
an unrelated space such as ${\cal H}_{\rm phys}$. Moreover, Dirac observables
will be replaced by the more accessible fashionables ${\cal F}$. (Since their
definition depends on the choice of time $Z$, and not just on the constraint
$C$, they are ``in fashion'' only as long as $Z$ is used as time. This need to
introduce fashionables has first been identified in semiclassical calculations
\cite{EffTime,EffTimeLong,EffTimeCosmo}.)

\section{Quantum dynamical reduction: Linear case}
\label{sec:MainResults}

In this section, we assume that the constraint is of a form $C=C_H$ such that
$[Z,C_H]=i\hbar\mathbf{1}$ for some $Z\in{\cal A}$. This case is close to
standard deparameterization, as shown by (\ref{OCflow}). (As we will discuss
in more detail in Section~\ref{sec:linearizing}, a constraint should, in
general, be factorized in order to make it deparameterizable, in the sense
that $C=NC_H$ with suitable $C_H=C_H^*$, $N \in \mathcal{A}$.)

\begin{defi} \label{def:Deparameterization} A quantum constraint $C_H\in \mathcal{A}$\ is \emph{deparameterized} by the clock $(Z, \mathcal{F})$\ if $[Z, C_H] = i\hbar \mathbf{1}$\ and the commutant of $Z$\ has the following properties:
\begin{enumerate}
\item $Z'+\mathcal{A}C_H = \mathcal{A}$;\label{def:Depar1}
\item $Z'\cap \mathcal{A}C_H = \{0\}$;\label{def:Depar2}
\item the set $Z'\cup \{C_H\}$\ algebraically generates $\mathcal{A}$.\label{def:Depar3}
\end{enumerate}
\end{defi}

The main result of this section is that deparameterization of a
constraint $C_H$\ by a clock $(Z, \mathcal{F})$\ is a realization of the
constrained quantization of $C_H$ as an unconstrained quantum mechanical
system with degrees of freedom given by the fashionable algebra $\mathcal{F}$\
evolving in time. The consistency of this construction is shown by
\begin{prop}\label{prop:Deparameterization}  If $C_H\in{\cal A}$ is
  deparameterized by $(Z,{\cal F})$, then there is a bijection between the
  constraint surface and the space of states on $Z'$, $\phi\colon \Gamma_{C_H}
  \rightarrow \Gamma_{Z'}$ such that:
\begin{enumerate}
\item All gauge flows mapped to $\Gamma_{Z'}$\ by
  $\phi$\ are transversal to the constant time surfaces $\Gamma_{Z'}|_{\pi_t}$, $t\in{\mathbb R}$. \label{prop:Depar2}
\item The flow generated by $C_H$\ itself is mapped to a dynamical flow on
  $\mathcal{F}$. \label{prop:Depar3} 
\item Each positive state on $\mathcal{F}$\ specifies a positive physical state
  $[\omega]_{C_H} \in \Gamma_{C_H}/\thicksim_{C_H}$.\label{prop:Depar4} 
\end{enumerate}
\end{prop}

Existence of $\phi$ follows immediately from conditions~\ref{def:Depar1}
and~\ref{def:Depar2} of Definition~\ref{def:Deparameterization}. Since $Z' +
\mathcal{A}C_H = \mathcal{A}$, any linear functional on $\mathcal{A}$ is
completely defined by its restrictions to $\mathcal{A}C_H$ and $Z'$. Since
$Z'\cap \mathcal{A}C_H = \{0\}$\ these two restrictions can be specified
independently. Any state in $\Gamma_{C_H}$\ vanishes on $\mathcal{A}C_H$\ and
defines a state on $\Gamma_{Z'}$\ by restriction to $Z'$. Conversely any state
on $Z'$\ can be uniquely extended to the entirety of $\mathcal{A}$\ by setting
the extension to vanish on $\mathcal{A}C_H$ (see equations~(\ref{eq:Phi}) and~(\ref{eq:PhiInv}) below). Deparameterization provides sections of the quantum gauge flow (Definition~\ref{def:Cequiv}) through constant-time surfaces $\Gamma_{Z'}|_{\pi_t}$. Transversality of these sections, defined and proven in Section~\ref{sec:Transversality}, is the local quantum analog of incisiveness of a classical gauge section of Definition~\ref{ClassSection}. Potential issues arising from the weaker local nature of transversality are discussed in Section~\ref{sec:PhysStates}. We prove properties 2 and 3 of the proposition in Sections~\ref{sec:NonRelTime} and~\ref{sec:Positivity}, respectively, and provide an explicit example of a deparameterizable constraint in
Section~\ref{sec:DeparamExample}.

\begin{rem} As shown in our discussion of cosmological
  models, condition~\ref{def:Depar3} of
  Definition~\ref{def:Deparameterization} is satisfied by the algebras of many
  physical examples of interest where $C_H$\ and kinematical observables are polynomial in an
  ``energy'' $E\in\mathcal{A}$. This property is only used in
  Section~\ref{sec:Transversality} specifically to ensure that all gauge flow
  generators can be projected to $Z'$ using finite power series. In other
  situations it may be possible to replace this condition with sufficiently
  strong assumptions on the topology of $\mathcal{A}$.
\end{rem}

\subsection{Proof of transversality} \label{sec:Transversality}

Let a quantum constraint $C_H$\ be deparameterized by a clock $(Z,
\mathcal{F})$. In this subsection we prove that all of the flows generated by
constraint elements $AC_H \in \mathcal{A}C_H$\ can be mapped to $\Gamma_{Z'}$,
and are separated by the constant-time surfaces of $Z$. According to
Definition~\ref{GaugeFixedStates}, a constant-time surface
$\Gamma_{Z'}|_{\pi_t}$ in $\Gamma_{Z'}$\ contains precisely those states on
$Z'$ that vanish on $\ker\pi_t$. A non-vanishing gauge flow that preserves the
values assigned by states to $\ker\pi_t$ is therefore tangent to
$\Gamma_{Z'}|_{\pi_t}$ and hence remains unresolved by fixing the value of the
clock. By contrast, a flow that is separated by the $\Gamma_{Z'}|_{\pi_t}$
either pierces the surface or vanishes at each of its points. The
corresponding geometrical picture suggests a notion of transversality.
\begin{defi} \label{def:Trans}

\begin{enumerate}
\item   A one-parameter family of states $\bar{\omega}_{\lambda} \in \Gamma_{Z'}$\ is \emph{transversal} to a constant time surface $\Gamma_{Z'}|_{\pi_t}$\ if for all $\lambda'$\ such that $\bar{\omega}_{\lambda'} \in \Gamma_{Z'}|_{\pi_t}$\ either $\left. \frac{\rm d}{\rm d \lambda} \bar{\omega}_{\lambda} (A) \right|_{\lambda = \lambda'} = 0, \ \forall A\in Z'$\ or there is some $B\in  \ker \pi_t$\ such that $\left. \frac{\rm d}{\rm d \lambda} \bar{\omega}_{\lambda} (B) \right|_{\lambda = \lambda'} \neq 0$. \label{def:TransStates}
\item A flow on $Z'$\ is \emph{transversal} to $\Gamma_{Z'}|_{\pi_t}$ if every one-parameter family of states generated by it is transversal to $\Gamma_{Z'}|_{\pi_t}$. \label{def:TransFlow}
\end{enumerate}
\end{defi}

To map gauge flows onto $\Gamma_{Z'}$, let us first explicitly define the
bijection between $\Gamma_{C_H}$\ and $\Gamma_{Z'}$. Taken together, the two
conditions of deparameterization $Z' \cap \mathcal{A}C_H = \{0\}$ and $Z' +
\mathcal{A}C_H = \mathcal{A}$ imply that every $A$ can be written as
\begin{equation} \label{eq:AlgebraicDecomp}
A = B + G C_H \ ,
\end{equation}
where $B \in Z'$ is unique and $G \in \mathcal{A}$\ is unique up to adding
terms that are annihilated by $C_H$\ multiplied on their right because $C_H$
is not a divisor of zero in a single constrained system.  We define the
forward map $\phi\colon \Gamma_{C_H} \to \Gamma_{Z'}$\ by restriction: for any
$\omega \in \Gamma_{C_H}$\ and $A\in Z'$
\begin{equation} \label{eq:Phi}
\phi(\omega)(A) = \omega(A) \ .
\end{equation}
We use decomposition~(\ref{eq:AlgebraicDecomp}) to define the inverse map: for any $A \in \mathcal{A}$\ and $\bar{\omega} \in \Gamma_{Z'}$, we have $A = B + G C_H$, and define
\begin{equation} \label{eq:PhiInv}
\phi^{-1}(\bar{\omega}) (A) = \bar{\omega} (B) \ .
\end{equation}
Since $B=0$\ for any $A \in \mathcal{A}C_H$\ by condition~\ref{def:Depar2} of
Definition~\ref{def:Deparameterization}, clearly $\phi^{-1}(\bar{\omega}) \in
\Gamma_{C_H}\subset \Gamma$. A flow is mapped from $\Gamma_{C_H}$\ to
$\Gamma_{Z'}$\ by mapping the one-parameter families of states it generates:
$\bar{\omega}_{\lambda} = \phi(\omega_{\lambda}$).

We can further iterate relation~(\ref{eq:AlgebraicDecomp}) expressing $G$\
as a sum of elements from $Z'$\ and $\mathcal{A}C_H$\ to get a second-order
expression $A = B+\tilde{B} C_H+\tilde{G} C_H^2$\ with $B, \tilde{B} \in Z'$\ and $\tilde{G}\ \in
\mathcal{A}$, and so on. In fact, since by condition~\ref{def:Depar3} of Definition~\ref{def:Deparameterization} the
set $Z'\cup\{C_H\}$\ algebraically generates $\mathcal{A}$, every
$A\in \mathcal{A}$\ can be written as a finite-order polynomial in $C_H$
\begin{equation} \label{eq:PolExpansionInC}
A = B_0 + B_1 C_H + B_2 C_H^2 + \ldots B_M C_H^M \ ,
\end{equation}
with $B_i \in Z'$ and $M\in{\mathbb N}$. 

Further, the adjoint action of $C_H$\ preserves the commutant of $Z$, since for any $A \in Z'$
\begin{equation}\label{eq:ZAdjointAction}
\left[ \left[A, C_H\right], Z\right] = \left[ \left[Z, C_H\right], A\right] +\left[ \left[A, Z\right], C_H\right] = 0 \ ,
\end{equation}
since $[Z, C_H]=i\hbar \mathbf{1}$ according to
Definition~\ref{def:Deparameterization}, and $[A, Z]=0$; therefore $[A,
C_H]\in Z'$.
\begin{lemma}
  For a constraint $C_H$\ deparameterized by a clock $(Z,\mathcal{F})$, every
  gauge flow mapped to $\Gamma_{Z'}$\ from $\Gamma_{C_H}$\ is determined by
  the derivation $\vec{D}_HF=\frac{1}{i\hbar} [F, C_H]$\ on $Z'$.
\end{lemma}
\begin{proof}
First, we note that for $F, B \in \mathcal{A}$
\[
[F, BC_H] = B (i\hbar\vec{D}_HF) + [F, B] C_H
\]
Iterating by replacing $B$\ with $BC_H$, we get for any integer $n \geq 1$,
\[
[F, BC_H^n] = \sum_{i=1}^n {n \choose i} (-1)^{i-1} B ((i\hbar)^i\vec{D}_H^i F)
C_H^{n-i} + [F, B] C_H^n= (-1)^{n-1} B ((i\hbar)^n\vec{D}_H^n F) + G C_H \ . 
\]
We have combined the terms proportional to $C_H$ using
\[
G=\sum_{i=1}^{n-1} {n \choose i} (-1)^{i-1} B ((i\hbar)^i\vec{D}_H^i F)
  C_H^{n-i-1}  + [F, B] C_H^{n-1}\,.
\]
They all vanish when evaluated on states in $\Gamma_{C_H}$. Using the above result
and writing $A\in \mathcal{A}$ as a
polynomial in $C_H$\ as in equation~(\ref{eq:PolExpansionInC}), we have
\begin{equation}\label{eq:CflowProj}
[F, AC_H] = \sum_{n=1}^{M+1} (-1)^{n-1} B_{n-1} ((i\hbar)^n\vec{D}_H^n F) + GC_H
\end{equation}
for some $G \in \mathcal{A}$ and $B_i\in Z'$. For any state $\omega \in
\Gamma_{C_H}$, we have $\omega(G C_H) = 0$, and hence the flows induced by
constraint elements $AC_H$ satisfy
\[
i\hbar \frac{{\rm d}}{{\rm d} \lambda} \omega_{\lambda} (F) = \omega_{\lambda}\left(
  [F, AC_H] \right) = \sum_{n=1}^{M+1} (-1)^{n-1} \omega_{\lambda} \left( B_{n-1}
  ((i\hbar)^n\vec{D}_H^n F) \right) \ ,
\]
For $F\in Z'$\ we also have $\vec{D}_H^n F \in Z'$, and the above expression
can be computed entirely through the states' restriction to $Z'$, and the
action of $\vec{D}_H$\ on $Z'$. Hence, for $\bar{\omega}_{\lambda} =
\phi(\omega_{\lambda})$
\begin{equation}\label{eq:CFlowBreaking}
 i\hbar \frac{{\rm d}}{{\rm d} \lambda} \bar{\omega}_{\lambda} \left(
       F \right)  = \bar{\omega}_{\lambda}
  \left(  \sum_{n=1}^{M+1} (-1)^{n-1} (i\hbar)^n B_{n-1}\vec{D}_H^n \left( 
      F \right) \right) \ . 
\end{equation}
\end{proof}

For the flow generated on $\Gamma_{Z'}$\ by $C_H$ itself,
$B_0=\mathbf{1}$ and $B_n=0$ for $n>1$ in (\ref{eq:CflowProj}). For
$F=\mathbf{1}$\ equation~(\ref{eq:CFlowBreaking}) implies that along the flow,
${\rm d} \bar{\omega}_{\lambda} \left( (Z-t\mathbf{1}) \right)/{\rm d}\lambda
= 1 \neq 0$\ for any $t\in \mathbb{R}$. Therefore this particular flow is transversal to all constant time surfaces $\Gamma_{Z'}|_{\pi_t}$. In fact, the property 
\begin{equation}\label{eq:DZparametrization}
\bar{\omega} \left( (\vec{D}_HZ) F \right) = \bar{\omega} \left( \mathbf{1} F \right) =\bar{\omega}(F) \ , 
\end{equation}
is sufficient to ensure that all gauge flows mapped to $\Gamma_{Z'}$\ are
transversal to constant time surfaces. To show this, we first establish a
useful result that holds for any $\vec{D}_HZ \in Z'$\ that commutes with all elements of $Z'$, and a pair of non-negative integers $i \leq n$: 
\begin{eqnarray*}
\vec{D}_H^i \left( Z-t \mathbf{1} \right)^n &=& \vec{D}_H^{i-1} \left( n(Z-t
  \mathbf{1})^{n-1} (\vec{D}_HZ) \right) 
\\
&=& \vec{D}_H^{i-2} \left( n(n-1)(Z-t \mathbf{1})^{n-2} (\vec{D}_HZ)^2 + n(Z-t
  \mathbf{1})^{n-1} (\vec{D}_H^2Z) \right) 
\\
&=& \cdots
\\
&=& (Z-t \mathbf{1})^{n-i}  {n \choose i} (\vec{D}_HZ)^i + (Z-t
\mathbf{1})^{n-i+1} \Bigl( \cdots \Bigr) \ . 
\end{eqnarray*}
Therefore, for $\bar{\omega} \in \Gamma_{Z'}|_{\pi_t}$, and for any $F_1, F_2 \in Z'$
\begin{equation} \label{eq:AuxResult}
\bar{\omega} \left[ F_1 \left( \vec{D}_H^i \left( Z-t \mathbf{1} \right)^n
  \right) F_2 \right] = \left\{ \begin{array}{cl} 0 & \mbox{for } i<n \ ;\\
    \bar{\omega} \left( (\vec{D}_HZ)^n F_1 F_2 \right) &\mbox{for }
    i=n \ . \end{array} \right. 
\end{equation}

\begin{lemma} \label{lm:TotalLinearGauge} If $C_H$\ is deparameterized by $(Z, \mathcal{F})$, the flow of every $AC_H\in \mathcal{A}C_H$\ mapped to $\Gamma_{Z'}$\ is transversal to every constant time surface $\Gamma_{Z'}|_{\pi_t}$\ of $Z$.
\end{lemma}

\begin{proof} Using~(\ref{eq:PolExpansionInC}) we write
\[
A = B_0 + B_1 C_H + B_2 C_H^2 + \ldots B_M C_H^M 
\]
with $B_i\in Z'$.  Suppose this flow is not transversal to
$\Gamma_{Z'}|_{\pi_t}$\ for some $t \in \mathbb{R}$. Then,
according to Definition~\ref{def:Trans}, there is a one-parameter family of
states $\bar{\omega}_{\lambda} \in \Gamma_{Z'}$\ generated by $AC_H$\
according to equation~(\ref{eq:CFlowBreaking}) with $\bar{\omega}_{\lambda_0}
\in \Gamma_{Z'}|_{\pi_t}$\ for some $\lambda_0 \in \mathbb{R}$, such that the flow of $AC_H$\ does not vanish at $\bar{\omega}_{\lambda_0}$, and for
any $F \in Z'$\ using~(\ref{eq:CFlowBreaking}),
\begin{eqnarray}
0 &=& \left. i\hbar \frac{{\rm d}}{{\rm d} \lambda} \bar{\omega}_{\lambda} \left(
    (Z-t\mathbf{1})^{M+1} F \right) \right|_{\lambda = \lambda_0} 
\nonumber\\
&=& \, \bar{\omega}_{\lambda_0} \left(  \sum_{n=1}^{M+1} (-1)^{n-1} (i\hbar)^n B_{n-1}\vec{D}_H^n
  \left( (Z-t \mathbf{1})^{M+1}  F \right) \right) 
\nonumber\\
&=&  \sum_{n=1}^{M+1} (-1)^{n-1} (i\hbar)^n \bar{\omega}_{\lambda_0} \left(  B_{n-1} \sum_{m=1}^n {n
    \choose m} \left(\vec{D}_H^m (Z-t \mathbf{1})^{M+1}\right)  \left(
    \vec{D}_H^{n-m} F \right) \right) 
\nonumber\\
&=&  \sum_{n=1}^{M+1} (-1)^{n-1} (i\hbar)^n \sum_{m=1}^n {n \choose m} \bar{\omega}_{\lambda_0}
\left(  B_{n-1} \left(\vec{D}_H^m (Z-t \mathbf{1})^{M+1}\right)  \left(
    \vec{D}_H^{n-m} F \right) \right) 
\nonumber\\
&=&  \sum_{n=1}^{M+1} (-1)^{n-1} (i\hbar)^n \sum_{m=1}^n {n \choose m} \delta_{m,\, M+1}
\bar{\omega}_{\lambda_0} \left(  B_{n-1} (\vec{D}_HZ)^{M+1} \left( \vec{D}_H^{n-m} F
  \right) \right) 
\nonumber\\
&=& (-1)^M (i\hbar)^{M+1} \bar{\omega}_{\lambda_0} \left( B_M \mathbf{1}^{M+1} F \right) = (-1)^M
(i\hbar )^{M+1} \bar{\omega}_{\lambda_0} \left( B_M F \right) \ ,  \label{omegaBF}
\end{eqnarray}
where the Kronecker delta comes directly from~(\ref{eq:AuxResult}). This implies that $\bar{\omega}_{\lambda_0} \left( B_M F \right) =0$ for all $F \in
Z'$. Iterating the argument, (\ref{eq:CFlowBreaking}) also implies
\vspace{-2pt}
\begin{eqnarray*}
0 &=& i\hbar \left. \frac{{\rm d}}{{\rm d} \lambda} \bar{\omega}_{\lambda} \left(
    (Z-t\mathbf{1})^M F \right) \right|_{\lambda = \lambda_0} 
\\
&=& \, \bar{\omega}_{\lambda_0} \left(  \sum_{n=1}^{M+1} (-1)^{n-1} (i\hbar)^n B_{n-1}\vec{D}_H^n
  \left( (Z-t \mathbf{1})^M  F \right) \right) 
\\
&=& \, \bar{\omega}_{\lambda_0} \left(  \sum_{n=1}^{M} (-1)^{n-1} (i\hbar)^n B_{n-1}\vec{D}_H^n
  \left( (Z-t \mathbf{1})^M  F \right) \right) \\
  &&+ (-1)^M(i\hbar)^{M+1} \bar{\omega}_{\lambda_0} \left( B_M
  \left[ \vec{D}_H^{M+1} \left((Z-t\mathbf{1})^M F \right) \right] \right) \
. 
\end{eqnarray*}
By (\ref{omegaBF}), the second term in the final expression is zero for any
$F \in Z'$, giving \vspace{-2pt}
\begin{eqnarray*}
0 &=& \, \bar{\omega}_{\lambda_0} \left(  \sum_{n=1}^{M} (-1)^{n-1} (i\hbar)^n B_{n-1}\vec{D}_H^n
  \left( (Z-t \mathbf{1})^M  F \right) \right)  
\\
&=&  \sum_{n=1}^{M} (-1)^{n-1} (i\hbar)^n \sum_{m=1}^n {n \choose m} \bar{\omega}_{\lambda_0} \left(
  B_{n-1} \left(\vec{D}_H^m (Z-t \mathbf{1})^{M}\right)  \left(
    \vec{D}_H^{n-m} F \right) \right) 
\\
&=& (-1)^{M-1} (i\hbar)^M\bar{\omega}_{\lambda_0} \left( B_{M-1} (\vec{D}_HZ)^{M} F \right) =
(-1)^{M-1} (i\hbar)^{M} \bar{\omega}_{\lambda_0} \left( B_{M-1} F \right) \ , 
\end{eqnarray*}
which implies $\bar{\omega}_{\lambda_0} \left( B_{M-1} F \right) =0$ for all $F \in
Z'$. Continuing in this way, we establish that $\bar{\omega}_{\lambda_0} \left( B_n F
\right) =0$ for all $n$. Therefore the flow must
completely vanish at $\bar{\omega}_{\lambda_0}$ since
\[
\left. i\hbar \frac{{\rm d}}{{\rm d} \lambda} \bar{\omega}_{\lambda} \left( F \right)
\right|_{\lambda = \lambda_0}  = \sum_{n=1}^{M+1} (-1)^{n-1} (i\hbar)^n \bar{\omega}
\left( B_{n-1} (\vec{D}_H^n F) \right) = 0 \ ,
\]
proving our claim by contradiction.
\end{proof}

%
%
%
%
%

\subsection{Deparameterized time evolution}\label{sec:NonRelTime}

In this subsection we prove that for a constraint $C_H$\ that is
deparameterized by $(Z, \mathcal{F})$, the adjoint action of $C_H$ projects to
a dynamical flow on the fashionable algebra $\mathcal{F}$. Later on, in
Section~\ref{sec:Positivity} we will link deparameterized dynamics to the
physical states and the Dirac observables of the constrained system.

Let $S_{C_H} (\lambda)$\ denote the flow induced on $\Gamma$ by the adjoint
action of $C_H$ on $\mathcal{A}$, which by Lemma~\ref{lem:GammaCinv} preserves
the constraint surface $\Gamma_{C_H}$. Using the map defined in~(\ref{eq:Phi})
and~(\ref{eq:PhiInv}) we can transfer this flow from $\Gamma_C$\ to a flow on
$\Gamma_{Z'}$\ via $\bar{S}_{C_H} (\lambda) = \phi \circ S_{C_H}(\lambda)
\circ \phi^{-1}$. Explicitly, $\bar{S}_{C_H} (\lambda)\phi(\omega) = \phi
(S_{C_H} (\lambda)\omega)$\ for any $\omega \in \Gamma_{C_H}$. Using this
relation, for any $\bar{\omega} \in \Gamma_{Z'}$ and any $A \in Z'$
\begin{equation} \label{eq:ProjectedCHFlow}
i\hbar \frac{{\rm d}}{{\rm d} \lambda} \left( \bar{S}_{C_H} (\lambda) \bar{\omega} (A) \right) =  \bar{S}_{C_H} (\lambda) \bar{\omega} \left( [A,  C_H] \right) \ ,
\end{equation}
where we implicitly use the fact that $[A, C_H] \in Z'$, which follows from  equation~(\ref{eq:ZAdjointAction}). Therefore, $\bar{S}_{C_H} (\lambda)$\ is generated simply by the adjoint action of $C_H$\ on the subalgebra $Z'$. Referring back to Definition~\ref{def:RelationalEvol} from Section~\ref{sec:QClocks} we have the following result.

\begin{lemma}\label{FlowLemma} Let $C_H$\ be deparameterized by $(Z, \mathcal{F})$, then the flow $\bar{S}_{C_H} (\lambda) = \phi \circ S_{C_H}(\lambda)  \circ \phi^{-1}$ generates time evolution of states relative to $Z$\ on $\Gamma_{Z'}$.
\end{lemma}

\begin{proof}
  Since $[\cdot,C_H]$\ is a $*$--compatible derivation on $\mathcal{A}$, it is
  $*$-compatible also on $Z'$.  All that needs to be shown is that the gauge
  flow $\bar{S}_{C_H}(\lambda)$ maps states from $\left. \Gamma_{Z'}
  \right|_{\pi_t}$\ to $\left. \Gamma_{Z'} \right|_{\pi_{t+\lambda}}$.

For convenience let us denote the one-parameter family of states $\bar{\omega}_{\lambda} :=\bar{S}_{C_H}(\lambda) \bar{\omega}$, where $\bar{\omega} \in \left. \Gamma_{Z'} \right|_{\pi_t}$\ for some fixed $t$. Following the method of Lemma~\ref{EvolutionLemma}, for each $A\in Z'$\ we define a function that varies along the flow $f_A(\lambda) = \bar{ \omega}_{\lambda} \left( (Z-(t+\lambda)\mathbf{1}) A \right)$. The state $\bar{\omega}_{\lambda}$ belongs to $\left. \Gamma_{Z'} \right|_{\pi_{t+\lambda}}$ if and only if $f_A(\lambda)=\bar{\omega}_{\lambda}\left( (Z-(t+\lambda)\mathbf{1}) A \right) = 0$ for all $A \in Z'$. Suppose  all of the functions $f_A(\lambda')= 0$\ for some $\lambda'$. We can compute their derivatives along the flow using equation~(\ref{eq:ProjectedCHFlow}):
\begin{eqnarray*}
\left.   \frac{{\rm d}f_A}{{\rm d}\lambda} \right|_{\lambda=\lambda'} &=& \left. \frac{{\rm d}}{{\rm d} \lambda} \bar{\omega}_{\lambda} \left( (Z-(t+\lambda)\mathbf{1}) A \right) \right|_{\lambda=\lambda'}
\\
&=& \left. \frac{{\rm d}}{{\rm d} \lambda} \bigl(  \bar{\omega}_{\lambda} \left( (Z-t \mathbf{1}) A \right) - \lambda  \bar{\omega}_{\lambda}(A) \bigr) \right|_{\lambda=\lambda'}
\\
  &=& \frac{1}{i\hbar} \bar{\omega}_{\lambda'} \left( \left[ (Z-t\mathbf{1})A,
      C_H \right] \right) - \frac{\lambda'}{i\hbar} \bar{\omega}_{\lambda'} \left(
    \left[ A, C_H \right] \right) -\bar{\omega}_{\lambda'} (A) 
  \\
  &=& \frac{1}{i\hbar} \bar{\omega}_{\lambda'} \left( (Z-t\mathbf{1}) \left[ A,
      C_H \right] \right) - \frac{\lambda'}{i\hbar} \bar{\omega}_{\lambda'} \left(
    \left[ A, C_H \right] \right) 
  \\
  &=& \frac{1}{i\hbar} \bar{\omega}_{\lambda'} \bigl( (Z-(t+\lambda')\mathbf{1})
  \left[ A, C_H \right] \bigr) =  \frac{1}{i\hbar} f_{\left[ A, C_H \right]} (\lambda') = 0\ . 
\end{eqnarray*}
The last equality follows since $ [A, C_H] \in Z'$. Furthermore, by our
initial conditions $f_A(0) = 0$\ for all $A\in \mathcal{A}$\ since
$\bar{\omega}_0 = \bar{\omega} \in \left. \Gamma_{Z'} \right|_{\pi_{t+0}}$. It
follows that $\{f_A(\lambda)=0, \mbox{ for all } \lambda \}_{A\in
  \mathcal{A}}$\ is the solution to the flow given by
equation~(\ref{eq:ProjectedCHFlow}) with $\bar{\omega}\in \left. \Gamma_{Z'}
\right|_{\pi_t}$. Thus at any point along the flow
\[
\bar{\omega}_{\lambda} \left(  (Z-(t+\lambda)\mathbf{1}) A \right) = 0 , \ \
{\rm for\ all\ } \ A \in Z' \ .
\]
Therefore $\bar{S}_{C_H}(\lambda) \bar{\omega} \in \left. \Gamma_{Z'}
\right|_{\pi_{t+\lambda}}$, as required.
\end{proof}
\begin{rem} The above lemma does not guarantee that for an arbitrary $\omega \in \Gamma_{C_H}$\ the family $\phi(S_{C_H} (\lambda) \omega)$\ is a time evolution relative to $Z$; we also need $\phi(\omega) \in \Gamma_{Z'}|_{\pi_t}$\ for some $t\in \mathbb{R}$.
\end{rem}

As discussed in Section~\ref{sec:QClocks}, time evolution of states $\bar{\omega}_t \in \Gamma_{Z'}$ relative to $Z$\ maps to a one-parameter family of states on the fashionable algebra $\tilde{\omega}_t=\psi_t(\bar{\omega}_t)$, by restriction of $\bar{\omega}_t$\ to $\mathcal{F}$: we have $\tilde{\omega}_t (F)=\bar{\omega}_t (F)$\ for all $F\in \mathcal{F}$. Furthermore, at each value $t\in \mathbb{R}$\ of the clock the adjoint action of $C_H$\ can be projected from $Z'$\ to $\mathcal{F}$\ using the map of equation~(\ref{eq:ZtoFmap})
\begin{equation} \label{eq:FashFlowGenerator}
\vec{D}^{\prime}_H (t) F := \nu_t \circ \pi_t \left(\vec{D}_H F \right) \ ,
\end{equation}
which is a $*$-compatible derivation on $\mathcal{F}$ thanks to the fact that $\nu_t$ and $\pi_t$ are $*$-homomorphisms. $\vec{D}^{\prime}_H(t)$\ is therefore a dynamical flow on $\mathcal{F}$\ according to Definition~\ref{def:Dflow}.

\begin{lemma} Let $\bar{\omega}_t = \bar{S}_{C_H} (t-t_0) \bar{\omega}$\ for
  some $t_0 \in \mathbb{R}$\ and $\bar{\omega} \in \Gamma_{Z'}|_{\pi_{t_0}}$,
  then the image $\tilde{\omega}_t = \psi_t(\bar{\omega}_t) \in
  \Gamma_{\mathcal{F}}$, where $\psi_t$ is defined in (\ref{psit}) is
  generated by the dynamical flow $\vec{D}^{\prime}_H(t)$. That is
\begin{equation}\label{FlowF}
\frac{{\rm d}}{{\rm d}t} \tilde{\omega}_t (F) = \tilde{\omega}_t
(\vec{D}^{\prime}_H(t) F) \ . 
\end{equation}
\end{lemma}
\begin{proof}
By Lemma~\ref{FlowLemma}, $\bar{\omega}_t\in\Gamma_{Z'}|_{\pi_{t}}$\ for all $t$. Using equation~(\ref{eq:ProjectedCHFlow}), for any $A \in Z'$
\[
\frac{{\rm d}}{{\rm d}t} \bar{\omega}_t (A) = \bar{\omega}_t \left( \frac{1}{i\hbar} \left[ A, C_H\right] \right) = \bar{\omega}_t (\vec{D}_H F) \ .
\]
Since  $\tilde{\omega}_t (F) = \psi_t(\bar{\omega}_t)(F) = \bar{\omega}_t(F)$\ for any $F \in \mathcal{F}$, we also have
\[
\frac{{\rm d}}{{\rm d}t} \tilde{\omega}_t (F) = \frac{{\rm d}}{{\rm d}t}
\bar{\omega}_t (F) = \bar{\omega}_t (\vec{D}_H F) \ .
\]
Because $\bar{\omega}_t \in \left. \Gamma_{Z'} \right|_{\pi_t}$, it assigns
the same value to all elements that belong to a given coset generated by the
ideal $(Z-t\mathbf{1})Z'$. By definition in
equation~(\ref{eq:FashFlowGenerator}), for any value of $t$,
$\vec{D}^{\prime}_H (t) F$\ and $\vec{D}_H F$ are in the same coset relative
to the ideal $(Z-t\mathbf{1})Z'$. Hence $\bar{\omega}_t (\vec{D}_H F) =
\bar{\omega}_t (\vec{D}^{\prime}_H(t) F)$, from which (\ref{FlowF}) follows.
\end{proof}

As demanded by property~\ref{prop:Depar3} of
Proposition~\ref{prop:Deparameterization}, the gauge flow of $C_H$\
deparameterized by $(Z, \mathcal{F})$\ projects to a dynamical flow on the
unconstrained system defined by the fashionable algebra $\mathcal{F}$. The
unconstrained dynamics on $\mathcal{F}$\ can then be analyzed by the methods
of conventional quantum mechanics.

\subsection{Positivity}
\label{sec:Positivity}

For many applications it is also important to be able to reverse this process,
linking positive states on a fashionable algebra at a given value of the clock
with kinematical solutions to the constraint in $\Gamma$ and ultimately with
physical states in $\Gamma_{C_H}/\thicksim_{C_H}$. Combining
Definition~\ref{def:Deparameterization} and Lemma~\ref{CosetLemma},
deparameterization is accomplished by finding a clock such that $[Z, C_H] =
i\hbar \mathbf{1}$, and by splitting the kinematical algebra into subalgebras
that share only the null element,
\[
\mathcal{A} =\mathcal{A}C_H + (Z-t\mathbf{1})Z' + \mathcal{F} \ ,
\]
where $\mathcal{F}$ is a $*$--subalgebra of $Z'$ isomorphic to
$Z'/(Z-t\mathbf{1})Z'$ at each $t$.

Specifically, given a positive state $\tilde{\omega}$\ on $\mathcal{F}$\ and a value of the clock $t\in \mathbb{R}$, the corresponding relational state is $\bar{\omega} = \psi_t^{-1} (\tilde{\omega}) \in \Gamma_{Z'}|_{\pi_t}$, which by Lemma~\ref{lem:RelationalPositivity} is positive on $Z'$. The corresponding solution of the constraint on $\mathcal{A}$\ is $\omega = \phi^{-1} ( \psi_t^{-1} (\tilde{\omega})) \in \Gamma_{C_H}$, which for any $B\in Z'$\ yields
\[
\omega( (Z-t\mathbf{1})B) = \bar{\omega}( (Z-t\mathbf{1})B) = 0 
\]
because $\bar{\omega}\in \Gamma_{Z'}|_{\pi_t}$.  This, in turn, implies that
\[
 \omega(Z) = t \quad\mbox{and}\quad \omega(ZB) = t\,
 \omega(B) = \omega(Z)\, \omega(B)\,.
\]
We find that this relation extends beyond $Z'$ to the whole of $\mathcal{A}$:
Given any $A\in \mathcal{A}$, and any state $\omega \in \Gamma_{C_H}$, that
satisfies $\omega(ZB) = \omega(Z)\, \omega(B)$ for all $B\in Z'$, there are $B
\in Z'$\ and $G \in \mathcal{A}$ such that
\begin{eqnarray*}
\omega(ZA) &=& \omega \left( Z B+ Z G C_H\right) 
\\
&=& \omega(ZB)
\\
&=& \omega(Z) \omega(B)
\\
&=& \omega(Z) \omega \left( B +  G C_H\right) = \omega(Z) \omega(A) \ . 
\end{eqnarray*}

The following definition is therefore meaningful:
\begin{defi}\label{def:AlmostPos}
A state $\omega \in \Gamma$ is {\em almost-positive} with respect to a
deparameterization of $C_H$ by $Z$ if 
\begin{enumerate}
\item it annihilates the left ideal generated by $C_H$: $\omega(AC_H)
  = 0$ for all $A \in \mathcal{A}$; 
\item it is positive on the commutant of $Z$: $\omega(BB^*) \geq 0$
  for all $B \in Z'$; 
\item it parameterizes left multiplication by $Z$: for all
  $A \in \mathcal{A}$, $\omega(ZA) = \omega(Z) \omega(A)$.  
\end{enumerate}
\end{defi}
The first condition ensures that $\omega \in \Gamma_{C_H}$ solves the
constraint. The second condition ensures that $\omega$ restricts to some
positive state on $Z'$, and hence also on $\mathcal{F}$. The third condition
ensures that this restriction belongs to the constant clock surface
$\left. \Gamma_{Z'}\right|_{\pi_{\omega(Z)}}$. In other words,
\begin{cor} \label{cor:Positive}
  Every positive state on $\mathcal{F}$\ has a unique extension to an
  almost-positive state $\omega$, for any given value of the clock
  $\omega(Z)=t\in\mathbb{R}$. Conversely, every almost positive state
  $\omega$\ restricts to a positive relational state on $Z'$\ and hence a
  positive state on $\mathcal{F}$ at time $\omega(Z)$.
\end{cor}

\begin{rem}
  One way to interpret the last condition in Definition~\ref{def:AlmostPos} is
  to notice that it requires fluctuations of $Z$ to vanish. For example, we
  have $(\Delta_{\omega} Z)^2:=\omega(Z^2)-\omega(Z)^2=0$. Just like a time
  parameter in ordinary quantum mechanics, $Z$ is sharply defined in such a
  state, but it does not correspond to an evolving observable since $Z$ is not
  an element of $\mathcal{F}$.  As already noted in Section~\ref{sec:QClocks}, the
  combination of almost-positivity and conditions required for
  deparameterization prevent $\omega$ from being positive on the whole of
  $\mathcal{A}$. The new notion of an almost-positive state, introduced here,
  may therefore be considered a maximal implementation of positivity in an
  internal-time setting, in which evolution is defined with respect to an
  algebra element. According to Lemma~\ref{lem:RelationalPositivity},
  positivity of states can be extended from observables (realized here by
  ${\cal F}$) to time ($Z\in Z'$), but not to the full algebra ${\cal A}$.
\end{rem}

According to Lemma~\ref{FlowLemma}, the gauge flow $\bar{S}_{C_H} (\lambda)$
generated by $C_H$ and mapped to $\Gamma_{Z'}$ drags states from the
subspace $\left. \Gamma_{Z'}\right|_{\pi_{t}}$ to the subspace
$\left. \Gamma_{Z'}\right|_{\pi_{t+\lambda}}$ of gauge-fixed states.
\begin{lemma} 
  The gauge flow $S_{C_H}$ on $\Gamma$ drags an almost-positive state
  $\omega$ to another almost-positive state $S_{C_H} (\lambda) \omega$, such
  that $S_{C_H} (\lambda) \omega(Z) = \omega(Z) + \lambda$. 
\end{lemma}
\begin{proof}
  Lemma~\ref{lem:GammaCinv} guarantees that the flow remains on
  $\Gamma_{C_H}$. Since the adjoint action of $C_H$ is a $*$-compatible
  derivation on $Z'$, according to Lemma~\ref{EvolutionLemma} the corresponding flow maps states that are positive on
  $Z'$ to other states that are positive on $Z'$. For any almost-positive $\omega$,
  along the flow $\omega_{\lambda}:=S_{C_H} (\lambda) \omega $ we have
\[
\frac{{\rm d}}{{\rm d}\lambda} \omega_{\lambda} (Z)  =  \omega_{\lambda}
(\vec{D}_H Z) =  1 \ . 
\]
Therefore, $\omega_{\lambda} (Z) = \omega(Z) + \lambda$.

To prove that parameterization of $Z$\ is preserved along the flow, we follow the method of Lemma~\ref{EvolutionLemma} and define a function $f_A(\lambda) = \omega_{\lambda}(ZA) - \omega_{\lambda}(Z)\omega_{\lambda}(A)$\ for each $A \in \mathcal{A}$. Condition~3 of Definition~\ref{def:AlmostPos} holds for $\omega_{\lambda}$\ if and only if $f_A(\lambda) = 0$\ for all $A\in \mathcal{A}$. Suppose all of the functions $f_A(\lambda') = 0$\ for some $\lambda'$. Taking an arbitrary $A\in\mathcal{A}$
\begin{eqnarray*}
\left.\frac{{\rm d}}{{\rm d}\lambda} \omega_{\lambda} (ZA) \right|_{\lambda=\lambda'}
&=& \left. \left( \omega_{\lambda} \left( Z \vec{D}_H A \right) +
    \omega_{\lambda} (A) \right)  \right|_{\lambda=\lambda'} 
\\
&=& \omega_{\lambda'} (Z)  \left. \frac{{\rm d}}{{\rm d}\lambda} \omega_{\lambda} (A)
\right|_{\lambda=\lambda'} + \left.  \frac{{\rm d}}{{\rm d}\lambda} \omega_{\lambda} (Z)
\right|_{\lambda=\lambda'} \omega_{\lambda'} (A) 
\\
&=& \left.  \frac{{\rm d}}{{\rm d}\lambda}  \left( \omega_{\lambda} (Z)
    \omega_{\lambda} 
    (A) \right) \right|_{\lambda=\lambda'} \ .
\end{eqnarray*}
Consequently, ${\rm d}f_A(\lambda)/{\rm d}\lambda = 0$\ at $\lambda=\lambda'$\
for all $f_A(\lambda)$. Since $f_A(0)=\omega(ZA) - \omega(Z)\omega(A)=0$ for
all $A\in \mathcal{A}$, it follows that $\{f_A(\lambda)=0, \forall \lambda
\}_{A\in \mathcal{A}}$\ is the solution to the flow generated by $C_H$. Hence
$\omega_{\lambda} (ZA) = \omega_{\lambda} (Z) \omega_{\lambda} (A)$ remains
true everywhere along the flow.
\end{proof}

Since ${\cal A}_{\rm obs}$ is not available in general, there is no full
quantum analog of Proposition~\ref{PropClass}. But any available Dirac
observable $O\in{\cal A}_{\rm obs}$ is a valid observable with respect to any
almost-positive state:
\begin{lemma}\label{lem:Positivity}
  If $O\in\mathcal{A}_{\rm obs}$, $\omega(OO^*)\geq 0$ for any almost-positive
  functional $\omega$ with respect to a deparameterization of $C_H$ by some
  $(Z, \mathcal{F})$.
\end{lemma}
\begin{proof}
  Since $O\in{\cal A}_{\rm obs}\subset{\cal A}$ is also an element of
  $\mathcal{A}$, the decomposition induced according to
  Definition~\ref{def:Deparameterization} by deparameterization implies that
  we can write it as $O = AC_H + B$ for some $A\in \mathcal{A}$ and $B\in
  Z'$. The fact that $O$ is in the commutant of $C_H$ implies $[O, C_H] = [B,
  C_H] + [A, C_H] C_H = 0$.  The second term on the left-hand side is in
  $\mathcal{A} C_H$, while equation~(\ref{eq:ZAdjointAction}) implies that the
  first term is in $Z'$. Since the two subalgebras are linearly independent,
  the two terms must vanish separately, implying that $[B, C_H] = 0$.  Because
  $C_H$ is not a divisor of zero in a single constrained system, for $A\in
  \mathcal{A}$ we have that $AC_H = 0$ implies $A=0$, and therefore $[A, C_H]
  = 0$.  These results also imply that $A^*$ and $B^*$ commute with $C_H$. Now
  suppose that $\omega$ is almost-positive with respect to deparameterization
  of $C_H$ by $Z$, then
\begin{eqnarray*}
\omega( OO^*) &=& \omega \left( AC_HC_HA^* + BC_HA^* +  AC_HB^* + BB^* \right)
\\
&=& \omega \left( \left( AA^*C_H + BA^* + AB^* \right) C_H \right) + \omega
\left(   BB^* \right) 
\\
&=& \omega \left( BB^* \right) \geq 0 
\end{eqnarray*}
because $B\in Z'$, using $\omega({\cal A}C)=0$.
\end{proof}
This completes the proof of Proposition~\ref{prop:Deparameterization}.

\subsection{Example: Parameterized Particle}\label{sec:DeparamExample}

Let ${\cal A}$ be the polynomial algebra, generated by complex polynomials in
the basic elements $Z$, $E$, $A_i$, with $i=1, 2, \ldots , M$, and
$\mathbf{1}$ ($=: A_0$). The generating elements are $*$-invariant, $Z=Z^*$,
$E=E^*$ and $A_i=A_i^*$, and are subject to commutation relations $[Z, E] =
i\hbar \mathbf{1}$, $[Z,A_i]=0=[E,A_i]$ and $[A_i, A_j] = i \hbar \,
\sum_{k=0}^{M} \alpha_{ij}^k A_k$ for some $\alpha_{ij}^k \in \mathbb{C}$.
This algebra is an example of an enveloping algebra of a Lie algebra, which we
equip with the $\rho$-topology of \cite{UnboundedTop}.

For any $M$-tuple of integers $\vec{n} =
(n_1, n_2, \ldots n_M)$, we define $A_{\vec{n}} =A_1^{n_1}A_2^{n_2} \ldots
A_M^{n_M}$, with $A_{\vec{0}} :=\mathbf{1}$. The set of monomials
$\{A_{\vec{n}} Z^mE^l \}$ is linearly independent and forms a linear basis on
$\mathcal{A}$. Let this system be subject to a single constraint of the form
\[
C_H=E+{\rm h}(Z,A_i)\ ,
\]
where ${\rm h}$ is a polynomial with an ordering such that ${\rm h}={\rm h}^*$ and therefore $C_H=C_H^*$. 

Consider the clock $(Z, \mathcal{F})$, where $\mathcal{F}$\ is the algebra generated by all complex polynomials in just the elements $A_i$, and is therefore spanned by the linear basis of monomials $\{ A_{\vec{n}}\}$. $Z=Z^*$\ is given, and, since $Z$\ commutes with itself and generators $A_i$, any element of $Z'$ can be written as a linear combination of monomials of the form $A_{\vec{n}} Z^m$. The relational observables of this clock are given by the projection $\pi_t(A)\in Z'/(Z-t\mathbf{1})Z'$ for any $A\in Z'$, interpreted as ``$A$ when $Z=t$.'' Here $\ker \pi_t =(Z-t\mathbf{1})Z'$, and it is easy to see that no non-zero elements of  $(Z-t\mathbf{1})Z'$\ are polynomials in generators $A_i$\ alone, therefore $\mathcal{F} \cap \ker \pi_t = \{0\}$\ as required by Definition~\ref{def:QClock}. For a basis monomial of $Z'$\ and any $t\in \mathbb{R}$
\begin{eqnarray*}
 A_{\vec{n}} Z^m &=&  A_{\vec{n}} \left( \left( Z -t\mathbf{1} \right)
   +t\mathbf{1}\right)^m \\ 
 &=& A_{\vec{n}} \sum_{k=0}^{m} {m\choose k}  \left( Z -t\mathbf{1} \right)^k
 t^{m-k}\\ 
 &=&  A_{\vec{n}} t^m +  A_{\vec{n}}  \sum_{k=1}^{m} {m\choose k}  \left( Z
   -t\mathbf{1} \right)^k t^{m-k} \ . 
\end{eqnarray*}
The last sum lies in the ideal $(Z-t\mathbf{1})Z'$, and, therefore, in the kernel of $\pi_t$; hence,
\[
\pi_t \left(  A_{\vec{n}} Z^m \right) =\pi_t \left(  A_{\vec{n}} t^m  \right)  = t^m \pi_t \left( A_{\vec{n}}  \right) \ . 
\]
Clearly, $\left\{ \pi_t \left( A_{\vec{n}}  \right) \right\}$\ linearly spans $ Z'/(Z-t\mathbf{1})Z'$\ for any $t\in \mathbb{R}$, so that $\pi_t (\mathcal{F}) =  Z'/(Z-t\mathbf{1})Z'$. The pair $(Z, \mathcal{F})$\ satisfy all requirements for a quantum clock given in Definition~\ref{def:QClock}.  

We now confirm that $(Z, \mathcal{F})$\ deparameterizes $C_H$\ according to Definition~\ref{def:Deparameterization}. We immediately see that $[Z, C_H] = [Z, E] = i\hbar \mathbf{1}$\ as required. Since the expression for any non-zero element of $\mathcal{A}C_H$ in terms of the basis monomials of $\mathcal{A}$\ includes at least one term of the form $A_{\vec{n}} Z^mE^l $ with $l\neq0$, we infer $Z' \cap \mathcal{A}C_H = \{0\}$.  By substituting $E = C_H - {\rm h}(Z, A_i)$, we can
write any element of $\mathcal{A}$ as a polynomial in $Z$, $C_H$, and
$A_i$. Using the commutation relations, a factor of $C_H$ can be iteratively
moved to the right whenever present, so that any $A \in \mathcal{A}$ can be
written as
\[
A = {\rm p_0} (Z, A_i) + {\rm p}(Z, C_H, A_i) C_H \ , 
\]
for some polynomials ${\rm p}_0$ and ${\rm p}$. The first term is in $Z'$,
while the second is in $\mathcal{A}C_H$, the two sets therefore linearly
generate the whole of $\mathcal{A}$. It is also immediately clear here that $Z'\cup\{C_H\}$\ algebraically generates $\mathcal{A}$: iteratively moving every factor of $C_H$\ to the right we can write 
\[
A = {\rm p_1} (Z, A_i) + {\rm p_2}(Z, A_i) C_H + \ldots + {\rm p}_N(Z, A_i) C_H^N \ .
\]

According to Proposition~\ref{prop:Deparameterization} the flow generated by $C_H$\ gives rise to a dynamical flow on the fashionables. To see this explicitly, we note that any $B\in Z'$ can be written as a polynomial ${\rm p}(Z, A_i)$, so that the adjoint action of $C_H$\ on $Z'$\ is given by
\[
[B, C_H] = [B, E] + [B, {\rm h}] = i\hbar \left. \frac{\partial {\rm p}(s,
    A_i)}{\partial s} \right|_{s=Z} + [{\rm p}(Z, A_i), {\rm h}(Z, A_i)] \ , 
\]
both terms in the final expression are in $Z'$. Since $[A_{\vec{n}}, E] = 0$, restricting this to $F \in \mathcal{F}$ we have
\[
\vec{D}_H F = \frac{1}{i\hbar} [F, C_H] =  \frac{1}{i\hbar} [F, {\rm h} (Z,
A_{\vec{n}}) ] \ . 
\]
The projection from $Z'$ to $\mathcal{F}$ here has the explicit form
\[
\nu_t\circ \pi_t \left( A_{\vec{n}} Z^m \right) = \nu_t \left( t^m \pi_t \left( A_{\vec{n}}\right) \right)= A_{\vec{n}} t^m \ , 
\]
so that the commutator of two basis monomials projects as
\[
\nu_t\circ \pi_t \left( \left[ A_{\vec{n}_1}, A_{\vec{n}_2} Z^m \right]
\right) =  \left[ A_{\vec{n}_1}, A_{\vec{n}_2} \right]  t^m \ . 
\]
It follows that
\[
\vec{D}^{\prime}_H (t) F =   \frac{1}{i\hbar} [F, {\rm h} (t,  A_{\vec{n}}) ] \ .
\]
Therefore, this deparameterized example reduces the original constrained
quantum system to an unconstrained quantum system with degrees of
freedom generated by the basis $\{ A_{\vec{n}} \}$, driven by the
time-dependent Hamiltonian ${\rm h} (t, A_{\vec{n}})$.

In this example, ${\cal A}_{\rm obs}=C_H'$ contains $C_H$ itself and the
identity element $\mathbf{1}$. Any element of ${\cal A}_{\rm obs}$ which is
not a linear combination of a power of $C_H$ and $\mathbf{1}$ is a constant of
motion of the (possibly time-dependent) Hamiltonian ${\rm h}$.  For most
choices of a classical polynomial ${\rm h}_{\rm class}$, constants of motion
which are $E$-independent and fulfill $\{O,h_{\rm class}\}=0$, are generically
non-polynomial, if they even exist in closed form
\cite{DiracChaos,DiracChaos2}. (The system may be non-integrable.) No
quantization of such an observable can exist in our ${\cal A}$, and the
available ${\cal A}_{\rm obs}$ is incomplete. Even if one extends ${\cal A}$,
for instance by using deformation quantization, in most cases of physical
interest it is impossible to find a complete set of Dirac
observables. Nevertheless, we have shown that it is possible to fix the gauge
relative to $Z$ in any such system and uniquely specify physical states by
relational observables, with the only requirement on ${\rm h}$ that ${\rm
  h}\in Z'$ and ${\rm h}^* = {\rm h}$.

\subsection{Physical vs.\ relational states} \label{sec:PhysStates}

In Section~\ref{sec:QConstraint} we defined physical states as orbits on
$\Gamma_{C_H}$\ of the entire collection of gauge flows generated by
$\mathcal{A}C_H$, such that they are positive on the Dirac observables
$\mathcal{A}_{\rm obs}$. In Section~\ref{sec:Positivity} we established that
relational states associated with a deparameterization of a constraint $C_H$\
by a clock $(Z, \mathcal{F})$\ are represented on $\Gamma_{C_H}$\ by almost
positive states of Definition~\ref{def:AlmostPos}.  According to
Lemma~\ref{lem:Positivity} each almost positive state is in the orbit
corresponding to some physical state of the constraint. Furthermore, since
time-evolution in $Z$\ is generated by $C_H$, it is tangential to the gauge
orbit so that all almost positive states along a time-evolution correspond to
the same physical state. Deparameterization of
Definition~\ref{def:Deparameterization} interprets orbits of physical states
(or, more accurately, a single preferred flow along these orbits) as time
evolution relative to a physical clock. This relation leaves room for two
important concerns.

First, some physical states may not be sampled by the deparameterization
relative to a given clock $Z$: It may happen that a physical orbit
$[\omega]_{C_H}$\ does not contain any state that parameterizes $Z$\ as in the
third condition of Definition~\ref{def:AlmostPos}. Stated differently, there
is no guarantee that $\phi \left( [\omega]_{C_H} \right) \cap \Gamma_{Z'}
|_{\pi_t} \neq \emptyset$\ for some $t\in \mathbb{R}$. While this concern
warrants further investigation, it does not immediately appear to be
physically problematic: since a clock is a physical part of the constrained
system, the ability to assign simultaneous definite values to the clock and
its commutant should in general require special states. At present, we do not
claim that every physical state can be captured by deparameterization relative
to a given clock.

Second, a given deparameterization may suffer from the analog of Gribov
problems in gauge theories if a physical state is sampled by multiple
relational states, in other words, if $\phi \left( [\omega]_{C_H} \right) \cap
\Gamma_{Z'} |_{\pi_t}$\ contains more than one state for some $t\in
\mathbb{R}$. If this happens, each of those states will give rise to a
time-evolution relative to $Z$, and a single physical state will be associated
with multiple inequivalent evolving states on $\mathcal{F}$. However only the
states that are \emph{positive} on $Z'$\ have a physically useful relational
interpretation (see discussion in Section~\ref{sec:QClocks}
and~\ref{sec:Positivity}). Therefore this possibility constitutes an ambiguity
of physical relational states only if $\phi \left( [\omega]_{C_H} \right) \cap
\Gamma_{Z'} |_{\pi_t}$\ contains more than one state that is positive on
$Z'$. An indication that such physical ambiguity is avoided by
deparameterization comes from the following consideration. One would expect
that multiple intersections of two $\Gamma_{Z'} |_{\pi_t}$\ by the orbit $\phi
\left( [\omega]_{C_H} \right)$\ would be generated by following a
$Z'$-positivity-preserving gauge flow. From equation~(\ref{eq:CFlowBreaking})
(together with (\ref{eq:PolExpansionInC})), we see that the only gauge flow
that, when mapped onto $\Gamma_{Z'}$, is generated by a $*$-compatible
derivation, and therefore preserves positivity on $Z'$, is the one generated
by $C_H$\ itself. For this particular flow we note the following result.

\begin{lemma}\label{lm:IncisiveLinearGauge} Let $C_H$\ be deparameterized by a clock $(Z, \mathcal{F})$. Then, each one-parameter family of states in $\Gamma_{Z'}$ along the constraint flow generated by $C_H$\ intersects each constant time surface at most once.
\end{lemma}
\begin{proof} Along the flow of $C_H$\ we have ${\rm
    d}\bar{\omega}_{\lambda} (Z)/{\rm d}\lambda=\bar{\omega}_{\lambda}(\vec{D}_HZ)= 1\not=0$. Therefore,
  $\bar{\omega}_{\lambda}(Z)$ is monotonic in $\lambda$, and any two states $\bar{\omega}_{\lambda_1}$\ and $\bar{\omega}_{\lambda_2}$ with $\lambda_1 \neq \lambda_2$\ along this flow assign different values to $Z-t\mathbf{1}\in\ker\pi_t$\ for any given $t\in \mathbb{R}$.
 It is, therefore, not possible for both states to belong to the same constant time surface $\Gamma_{Z'}|_{\pi_t}$.
\end{proof}

\section{General polynomial constraints} \label{sec:linearizing}

A general constraint element $C$ is not immediately of the form required for
deparameterization to exist. While the kinematical algebraic structure of most
model theories that are of interest to quantum gravity and quantum cosmology
possesses clocks $Z$ that are not constant along the constraint flow, that is,
$[Z,C]\not=0$, the condition $[Z, C]=i\hbar \mathbf{1}$\ (or
$[Z,C]=a\mathbf{1}$ with $a\in{\mathbb C}$) is not usually satisfied. For
instance, most Hamiltonian constraints in such systems are quadratic in
momenta, resulting in $[Z, [Z, C]] \neq 0$ for a clock $Z$, as in the scalar
example of cosmological models with energy density (\ref{ScalarDens}). There
are also examples of constraints for which $[Z,C]\in Z'$ but not a multiple of
the unit.

\subsection{Linearization and cancellation}

In some of these cases, it may be possible to ``linearize'' the constraint by
finding a suitable $C_H\in{\cal A}$ which satisfies all three criteria of a
deparameterization with respect to $Z$ and has a gauge flow and a constraint
surface related to those of $C$: If $N\in{\cal A}$ is such that $C= NC_H$ and
$C_H$ as in Definition~\ref{def:Deparameterization}, we have ${\cal A}_C\subset {\cal
  A}_{C_H}$ and therefore $\Gamma_{C_H}\subset \Gamma_C$. Moreover,
$\omega\thicksim_C\psi$ if $\omega\thicksim_{C_H}\psi$. If $N$ is invertible
in ${\cal A}$, $\Gamma_{C_H}=\Gamma_C$ and $\thicksim_{C_H}=\thicksim_C$, but
this case is too restrictive for most practical purposes.
\begin{ex}
  The constraint $C=E^2-{\rm h}(A_i)^2$ with $Z$-independent ${\rm h}={\rm h}^*$ can be
  factorized as $C=(E-{\rm h})(E+{\rm h})=NC_H$ with $N=E-{\rm h}$ and
  $C_H=E+{\rm h}$. We have $[N,C_H]=0$, but $N$ does not have an
  inverse. However, if $\omega\in\Gamma_{C_H}$ and
  $\omega(E)\not=0$\ it follows that $\omega(N)\not=0$, in which case it may be of interest to study evolution of
  the state with respect to $C_H$.  More generally, we define $C_{\pm} = E \pm
  {\rm h}( A_i)$ such that
\[
C = C_{+}C_{-} = C_{-}C_{+} \ .
\]
\end{ex}
Since either $\omega(AC_{+})=0$\ or $\omega(AC_{-})=0$\ also imply
$\omega(AC)=0$, every left solution of $C_{\pm}$\ is also a left solution of
$C$. Therefore both constraint surfaces $\Gamma_{C\pm}$ are  contained
within the constraint surface $\Gamma_C$. Furthermore, normalized combinations of
states from $\Gamma_{C+}$\ and $\Gamma_{C-}$\ also give us solutions to
$C$. In particular, for any $a_i^{(+)}, a_i^{(-)} \in \mathbb{C}$,
$\omega_i^{(+)} \in \Gamma_{C+}$, and $\omega_i^{(-)} \in \Gamma_{C-}$
\[
\omega = \sum_i a_i^{(+)} \omega_i^{(+)} + \sum_i a_i^{(-)} \omega_i^{(-)}
\]
is a left solution of $C$, which is normalized provided that $\left( \sum_i
  a_i^{(+)} + \sum_i a_i^{(-)} \right) =1$. 

In this example the two subsets $\Gamma_{C\pm}$ are not disjoint. A solution
to both constraint factors must satisfy $\omega(AC_+) = 0$ and $\omega(AC_-)
= 0$ for any $A \in \mathcal{A}$. These conditions are equivalent to
requiring that both $\omega(AE)=0$ and $\omega(A{\rm h}) = 0$ for all $A \in
\mathcal{A}$, since
\[
\omega(AE) = \omega \left(A\frac{1}{2} (C_+ +C_-) \right) = \frac{1}{2} \left(
  \omega(AC_+) + \omega(AC_-) \right) = 0 \ , 
\]
and similarly
\[
\omega(A{\rm h}) = \omega \left(A\frac{1}{2} (C_+ -C_-) \right) = \frac{1}{2}
\left( \omega(AC_+) - \omega(AC_-) \right) = 0 \ . 
\]
Conversely, $\omega(AE)=0$ and $\omega(A{\rm h}) = 0$ imply $\omega(AC_+) =
0$\ and $\omega(AC_-) = 0$. The only restriction on the values assigned by a
kinematical state $\omega \in \Gamma$, in addition to continuity, is
normalization $\omega(\mathbf{1})=1$. It is therefore possible to satisfy both
$\omega(AE)=0$ and $\omega(A{\rm h}) = 0$ for all $A$, unless $AE + B{\rm h} =
\mathbf{1}$ for some $A, B \in \mathcal{A}$. No such $A$ and $B$ exist within
our polynomial $\mathcal{A}$, hence the intersection $\Gamma_{C+}\cap
\Gamma_{C-}$ is non-empty. However if we consider only the states that are
positive on $Z'$, as required by Definition~\ref{def:AlmostPos}, there may be
additional restrictions: suppose ${\rm h} = FF^* + \epsilon_0 \mathbf{1}$ for
some $F\in \mathcal{F}$ and real $\epsilon_0 >0$.  Then, for any normalized
state that is positive on $\mathcal{F}$
\[
\omega({\rm h}) = \omega\left( FF^*\right) + \epsilon_0 \geq  \epsilon_0 > 0 \ ,
\]
which means $\omega({\rm h}) =0$ cannot be satisfied. Hence depending on
${\rm h}$ the sets of almost-positive states with respect to internal clock
$Z$ defined by the two constraint factors may be disjoint.

Using the original constraint $C$, the orbits are generated by the subalgebra
$\mathcal{A}C$, as opposed to $\mathcal{A}C_{\pm}$ if we use one of the
factors instead. Neither of the two factors has an inverse already contained
within $\mathcal{A}$ (the only element with an explicit inverse in
$\mathcal{A}$ here is $\mathbf{1}$). Thus $\mathcal{A}C$ is a proper subset
of $\mathcal{A}C_{\pm}$, and hence the original orbits of $C$ are contained
within the larger orbits of $C_{\pm}$. This guarantees, via
Lemma~\ref{lem:GammaCinv}, that a physical state of the original constraint
$C$ is either entirely inside $\Gamma_{C\pm}$ or entirely outside of it.
However, some gauge flows generated by the factor constraints $C_{\pm}$ are
not gauge orbits of $C$ and can potentially link distinct gauge orbits of the
original constraint $C$. Therefore, a physical state with respect to
$C_{\pm}$ generally corresponds to a region of the space of physical
states with respect to the original constraint $C$.

This complication would not arise if $N$\ had an inverse in
$\mathcal{A}$. However, even if $N$\ is non-invertible there are in general
some states on which its action can be ``reversed'' in the following sense.
\begin{defi}
  Left multiplication of $A\in \mathcal{A}$ can be {\em canceled in $\omega
    \in \Gamma$} if for any $B \in \mathcal{A}$, $\omega(GAB) = 0$ for all
  $G\in \mathcal{A}$ implies $\omega(GB) = 0$ for all $G\in
  \mathcal{A}$.
\end{defi}
This state-by-state condition differs from the algebraic cancellation
property. In our concrete example, only the zero element is a divisor of zero
in $\mathcal{A}$. However, even though $CB=0$\ implies $B=0$, left
multiplication by $C$ cannot be canceled in any of its left solutions $\omega
\in \Gamma_C$. Setting $B=\mathbf{1}$, we get
$\omega(GC\mathbf{1})=\omega(GC)=0$ for all $G \in \mathcal{A}$, which is not
equivalent to $\omega(G\mathbf{1}) = 0$ for all $G\in \mathcal{A}$, since
setting $G = \mathbf{1}$ would violate normalization.

\begin{defi}
  A constraint $C$ is {\em deparameterized by factorization} with respect to
  an internal clock $Z$, if there are $N, C_H \in \mathcal{A}$, such that (i)
  $C = N C_H$, (ii) there is at least one state $\omega \in \Gamma_{C_H}$\ in
  which left multiplication by $N$\ can be canceled, and (iii) $C_H=C_H^*$ is
  deparameterized by $(Z, \mathcal{F})$.
\end{defi}

In our concrete example, if we deparameterize our system using $C_+$ as the constraint, we consider
only the states $\omega \in \Gamma_{C_+}$ in which the left multiplication of
$C_-$ can be canceled. (In particular, this means that $\omega \notin
\Gamma_{C_-}$.)

\begin{lemma}
For a constraint that is deparameterized by factorization as $C=NC_H$, for any $A \in \mathcal{A}_{\rm obs}$ of $C$ and $\omega\in\Gamma_{C_H}$ such
that left multiplication by $N$\ can be canceled in $\omega$,
the value $\omega(A)$ is invariant along all
of the gauge flows generated by $C_H$.
\end{lemma}
\begin{proof}
We first observe that, since $[A, C]=0$, in
particular $\omega(B[A, C]) = 0$ for any $B\in \mathcal{A}$. Which means
\begin{eqnarray*}
0 &=& \omega \left( B[A, NC_H] \right) 
\\ 
&=&  \omega \left( BN[A,C_H] \right) + \omega \left( B[A, N]C_H \right) =
\omega \left( BN[A,C_H] \right) \ . 
\end{eqnarray*}
Since this holds for any $B$, cancellation of left multiplication by $N$
in $\omega$ implies that
\[
\omega \left( B[A,C_H] \right) = 0 \ , \quad {\rm for\ all\ }\ B\in
\mathcal{A} \ . 
\]
The above property holds along all of the gauge flows generated by $C_H$. To see this let us fix an arbitrary $G\in \mathcal{A}$\ and, following the method of Lemma~\ref{EvolutionLemma}, define functions $f_B(\lambda) = S_{GC_H} \omega \left( B[A,C_H] \right)$\ for each $B\in \mathcal{A}$. Clearly, $f_B(0)=0$\ for all $B\in \mathcal{A}$. Suppose all functions $f_B(\lambda')=0$\ for some $\lambda'$, then
\begin{eqnarray*}
\left. i\hbar \frac{{\rm d}f_B}{{\rm d} \lambda} \right|_{\lambda = \lambda'} &=& \left. i\hbar \frac{{\rm d}}{{\rm d} \lambda} \left( S_{GC_H}(\lambda) \omega\left( B[A,C_H] \right)  \right) \right|_{\lambda = \lambda'}
\\
&=& S_{GC_H}(\lambda') \omega\left( B[A,C_H] GC_H - GC_H B[A,C_H] \right)
\\
&=& S_{GC_H}(\lambda') \omega\left( B[A,C_H] GC_H\right) - f_{GC_HB} (\lambda') = 0 \ ,
\end{eqnarray*}
where we used the fact that $ S_{GC_H}(\lambda') \omega \in \Gamma_{C_H}$\ according to Lemma~\ref{lem:GammaCinv}. We see that $\{f_B(\lambda)=0, \forall \lambda \}$\ is the solution of the dynamical flow generated by $GC_H$\ that agrees with our initial conditions. Since $G$\ was arbitrary, $S_{GC_H}(\lambda) \omega \left( B[A,C_H] \right) = 0$\ for all $B, G\in \mathcal{A}$\ and $\lambda \in \mathbb{R}$.

Using the above result, the value of $\omega(A)$ along the gauge flow generated by $BC_H$, with arbitrary $B\in \mathcal{A}$, varies according to
\begin{eqnarray*}
i\hbar \frac{{\rm d}}{{\rm d} \lambda} \left( S_{BC_H}(\lambda) \omega(A) \right)  &=&   S_{BC_H}(\lambda) \omega([A, BC_H])
\\ 
&=& S_{BC_H}(\lambda)\omega\left( B[A, C_H] \right) + S_{BC_H}(\lambda)\omega\left( [A, B] C_H \right)= 0 \ , \ {\rm for\ all\ }\ \lambda \ .
\end{eqnarray*}
\end{proof}

Therefore, using gauge flows generated by $C_H$ does not affect the values
assigned to the set of Dirac observables of the original constraint $C$, so long as we use states on which left multiplication of the factor $N$ can be canceled. In this section's example, the
roles of $C_+$ and $C_-$ can be reversed since the two factors commute.

In principle this construction also applies to a constraint that can be written as a
product of non-commuting factors, as one would expect in the case of time-dependent
Hamiltonians ${\rm h}(A_i,Z)$. However, factorizing such a constraint is much more complicated.
\begin{ex}If we factorize a constraint of the form $C=E^2-H^2$ with
  $[E,H]\not=0$, we have $C\not=(E-H)(E+H)$, but we can try to find $X\in Z'$ such that $C= (E-H+X)(E+H-X)$. Multiplying the two factors, we have
\[
C=E^2-H^2-X^2+[E,H]-[E,X]+[X,H]+2HX
\]
provided that
\[
 2HX= [H,E]+[H,X]+[E,X]+X^2\,.
\]
This equation has a formal power-series solution $X=\sum_{n=1}^{\infty}
(i\hbar)^nX_n$ with
\[
 2HX_1 = \frac{[H,E]}{i\hbar}
\]
and
\[
 2HX_n=\frac{[H,X_{n-1}]}{i\hbar}+ \frac{[E,X_{n-1}]}{i\hbar}+
 \sum_{a=1}^{n-1} X_{n-a}X_a\,.
\]
We can split $X=\frac{1}{2}(X_++X_-)$ into its $*$-invariant and
anti-$*$-invariant contributions, $X_+=\frac{1}{2}(X+X^*)$ and
$X_-=\frac{1}{2}(X-X^*)$, and define
\[
 H'=H-X_+ \quad\mbox{and}\quad E'=E-X_-\,.
\]
As in the example with commuting factors, $H'{}^*=H'$ and
$[Z,E']=[Z,E]=i\hbar\mathbf{1}$ but $E'{}^*\not=E'$. There are therefore
almost-positive states, but the gauge flow of $C_H=E'+H'$ does not induce a
$*$-compatible derivation, unless it so happens that $X_- =0$.
\end{ex}

For a systematic analysis of suitable factorizations, we need to carefully
consider the adjointness conditions imposed on the factors of the constraint.

\subsection{Adjointness relations}

The adjointness relation $C^*=C$ imposed on constraints guarantees that ${\cal
  A}_{\rm obs}\subset{\cal A}$\ inherits a $*$-relation, which in turn makes
it possible to define physical states as positive linear functionals on ${\cal
  A}_{\rm obs}$. This condition also restricts possible factorization choices
that could be applied to linearize constraints. Suppose a constraint $C=C^*$
can be written as $C=NC_H$, where $N$ can be algebraically canceled within
$\mathcal{A}$ and $C_H=C_H^*$ allows a deparameterization with respect to
$Z^*=Z\in{\cal A}$. Then $C$ can be deparameterized with respect to $Z$ by
factorization, using the same method as we applied to $C=C_-C_+$ to cast a
subset of its physical states as dynamical evolution in $Z$. Under these
conditions, $C_H$ uniformizes the flow generated by $C$: Since $[Z,C]=i\hbar
N$, we may consider $N$ as the ``non-constant rate'' of evolution determined
by $C$, while evolution with respect to $C_H$ has constant rate.

In order to satisfy $C=C^*$ we need $NC_H = C_HN^*$, which can be rewritten as
\begin{equation} \label{eq:Ncondition}
[N, C_H] =  C_H(N^* - N) \ .
\end{equation}
If the non-constant rate is required to be real when evaluated in a positive
state $\omega$, we need $N^*=N$. In this case, (\ref{eq:Ncondition}) implies
$[N,C_H]=0$, such that the rate is, in fact, constant on solutions of the
constraint because $N$ is a constant of motion with respect to $C_H$.
Conversely, if $[N, C_H]=0$ we obtain $(N-N^*)C_H=0$, and if $C_H$ can be
algebraically canceled within $\mathcal{A}$, we get $N=N^*$. These cases
constitute two \emph{sufficient} conditions for factorization to result in a
deparameterization.

Provided that the clock is part of a canonical pair, $[Z, E] =
i\hbar\mathbf{1}$, as in the example from the previous subsection, the most
general form of a factorizable constraint is $C=N(E+H)$, where $H=H^*$
commutes with $Z$, and condition~(\ref{eq:Ncondition}) holds for $C_H =
E+H$. Further properties depend on the $E$-dependence of $C$.

\subsubsection{Non-relativistic constraints}

\begin{defi}
  A constraint $C\in{\cal A}$ is {\em non-relativistic of rate} $N\in{\cal A}$
  if there is a canonical generator $E\in{\cal A}$ conjugate to $Z\in{\cal
    A}$, $[Z,E]=i\hbar\mathbf{1}$, such that $[Z,C]=i\hbar N\in Z'$.
\end{defi}

\begin{defi}
  A non-relativistic constraint $C\in{\cal A}$ is {\em of constant flow rate}
  $N\in{\cal A}$ if there is a $C_H\in{\cal A}$ such that $C=NC_H$ and
  $[N,C_H]=0$.
\end{defi}

\begin{lemma} \label{lem:Flow}
 Every deparameterizable non-relativistic constraint is of constant flow rate.
\end{lemma}
\begin{proof}
  Since $C^*=C$ and $Z^*=Z$ imply $[Z, C]^*=-[Z, C]$, we immediately obtain
  $N^*=N$ from $N=[Z,C]/(i\hbar)$.  Using this in (\ref{eq:Ncondition}), we
  have $[N, C_H] = 0$.
\end{proof}
\begin{rem}
  The condition $[N, C_H] = 0$ of constant flow rate shows the restrictive
  nature of adjointness conditions: Only constants of motion with respect to
  $C_H$ are allowed as factors of $E$ in non-relativistic constraints. Written
  as $[N, E] = -[N, H]$ if $C_H=E+H$, the condition amounts to a partial
  differential equation for $N$ as a function of $Z$ and the remaining
  canonical variables.
\end{rem}

\begin{lemma}
  If a non-relativistic constraint $C$ is deparameterizable, it is of the form
  $C = A_1 E + A_0$ such that $A_1=A_1^*$ and $[A_0,A_1]=A_1[A_1,E]$.
\end{lemma}
\begin{proof}
  Since the constraint is non-relativistic, it is linear in $E$ and can be
  written as $C = A_1 E + A_0$ with $A_1$ and $A_0$ such that $[Z, A_1] = [Z,
  A_0] = 0$. The conditions $C=C^*$ and $Z=Z^*$ imply that $[Z, C]^*=-[Z, C]$,
  and thus $A_1=A_1^*$. Since $A_1$ plays the role of the factor $N$, it must
  be a left factor of $A_0$: There must exist $H\in \mathcal{A}$ such that
  $A_0 = A_1 H$ and
\[
C = A_1 (E+H) \ .
\]
In order for $C$ to be deparameterizable, according to Lemma~\ref{lem:Flow},
the flow rate $A_1=N$ must be constant with respect to $C_H=E+H$. Therefore,
$[A_1, E] = -[A_1, H]$, which, upon left multiplication with $A_1$, implies
$A_1[A_1, E ]=-[A_1,A_0]$ because $A_0=A_1H$.
\end{proof}
\begin{rem}
  The condition $H=H^*$, obtained from $C_H^*=C_H$ for a deparameterizable
  constraint, implies
\[
 A_0^*= HA_1=A_0+[A_1,H]\,.
\]
Therefore, $A_0$ in $C=A_1E+A_0$ is not self-adoint unless $A_1$ commutes with
$H$.
\end{rem}

\begin{rem}
  If, in spite of Lemma~\ref{lem:Flow}, we try to factorize a constraint of
  non-constant flow rate, we end up with a non-self-adjoint $C_H$. To see
  this, consider a non-relativistic constraint of the form
  $C=\frac{1}{2}(B_1E+EB_1)+B_0$ with $B_0=B_0^*$\ and invertible $B_1=B_1^*$,
  we can write
\[
 C=B_1E+B_0-\frac{1}{2}[B_1,E]= B_1\left(E+\frac{1}{2}(B_1^{-1}B_0+B_0B_1^{-1})+
   \frac{1}{2}[B_1^{-1},B_0]-\frac{1}{2}B_1^{-1}[B_1,E]\right)\,.
\]
Defining $N=B_1$,
\[
 H=\frac{1}{2}(B_1^{-1}B_0+B_0B_1^{-1})
\]
and
\[
 E'=E+\frac{1}{2}[B_1^{-1},B_0]-\frac{1}{2}B_1^{-1}[B_1,E]\,,
\]
we can write $C=NC_H$ with $C_H=E'+H$.  It follows that $H=H^*$, and
$[Z,E']=[Z,E]=i\hbar$ since $[Z,B_1]=0$ and $[Z,B_0]=0$ for a non-relativistic
constraint. However,
\[
 E'{}^*= E-\frac{1}{2}[B_1^{-1},B_0]+ \frac{1}{2}[B_1,E]B_1^{-1}\not=E'
\]
and therefore $C_H^*\not=C_H$.  For $\omega\in\Gamma_{C_H}$, we have
$\omega(E')=-\omega(H)\in{\mathbb R}$. If $\omega$ is almost-positive, this
equation is consistent even though $E'\not=E'{}^*$ while $H^*=H$: because
$E'\not\in Z'$, an almost-positive state may take on a real value in a
non-self-adjoint $E'$.  However, the gauge flow of $C_H\not=C_H^*$ does not
induce a $*$-preserving derivation on any fashionable algebra ${\cal F}\subset
Z'$ because, in general, $[f,E']\not=0$ for $f\in{\cal F}$ unless $B_0$ and
$B_1$ are multiples of the unit.
\end{rem}

\subsubsection{Relativistic constraints}

\begin{defi}
  A constraint $C$ is {\em relativistic} if there is a canonical generator $E$
  conjugate to $Z$, $[Z,E]=i\hbar\mathbf{1}$, such that
  $0\not=[Z,[Z,C]]\in Z'$.
\end{defi}

A relativistic constraint that is deparameterizable by factorization has the
form
\[
C = (N_1E+N_0)(E+H) = N_1 E^2 + (N_0 + N_1H)E+N_1 [E, H] + N_0H 
\]
where $H^*=H$, and $[Z, N_1] = [Z, N_0] = [Z, H] = 0$. Using $C^*=C$ and
$Z=Z^*$, we have
\[
\left( \frac{1}{i\hbar}\left[ Z, \frac{1}{i\hbar} [Z, C] \right] \right) =
\left( \frac{1}{i\hbar}\left[ Z, \frac{1}{i\hbar} [Z, C] \right] \right)^* \ , 
\]
which quickly yields $N_1 = N_1^*$. The flow rate of $C$ with respect to
$C_H=E+H$ is given by $N_1E+N_0$, such that $C=NC_H$. In contrast to linear or
relativistic constraints, the flow rate depends on $E$.

\begin{lemma} \label{lem:RelFact} If a relativistic constraint $C$\ that is
  deparameterizable by factorization is of constant real flow rate $N$, it is
  of the form $C=NC_H$ with $N=N_1E+N_0$ and $C_H=E+H$ such that
\begin{equation} \label{eq:RelCond1}
 N_0^*=N_0+[N_1,E]\,,
\end{equation}
\begin{equation} \label{eq:RelCond2}
[N_1, E] + [N_1, H]  = 0 \ ,
\end{equation}
and
\begin{equation} \label{eq:RelCond3}
N_1 [H, E] = [N_0, E ] + [N_0, H]
\end{equation}
\end{lemma}
\begin{proof}
For real flow rate, $N=N^*$ implies
$N^*=EN_1+N_0^*=N$ and therefore (\ref{eq:RelCond1}).
Constant flow rate, $[N, C_H] = 0$, results in
\begin{equation}\label{eq}
0 = \left( [N_1, E] + [N_1, H] \right) E + N_1 [E, H] + [N_0, E ] +[N_0, H] \ .
\end{equation}
Taking a commutator with $Z$ on both sides, only the term proportional to $E$
survives giving us (\ref{eq:RelCond2}). Substituting this back into (\ref{eq})
results in (\ref{eq:RelCond3}).
\end{proof}

The three conditions of Lemma~\ref{lem:RelFact} together are sufficient to make
the quadratic constraint deparameterizable by factorization.

\begin{lemma}
A relativistic constraint with real constant flow rate $N=E+N_0$ is of the form
$C=\tilde{E}^2-h$ such that $\tilde{E}^*=\tilde{E}$ and $h^*=h$ as well as
$[Z,\tilde{E}]=i\hbar\mathbf{1}$ and $[Z,h]=0$.
\end{lemma}
\begin{proof}
  A relativistic constraint with flow rate $N=E+N_0$, using $N_1=\mathbf{1}$
  in terms Lemma~\ref{lem:RelFact}, can be written as $C= E^2 + A_1 E+ A_0$,
  where $[A_i, Z] = 0$.  Using the notation of Lemma~\ref{lem:RelFact},
\begin{equation} \label{A1A0}
 A_1=N_0+H \quad\mbox{and}\quad A_0=[E,H]+N_0H\,.
\end{equation}
We have $A_1^*=A_1$ because $N_0^*=N_0$ from
equation~(\ref{eq:RelCond1}). Equation~(\ref{eq:RelCond2}) is trivially
satisfied, while~(\ref{eq:RelCond3}) becomes
\begin{equation} \label{EH}
 [H, E] = [N_0, E ] +[N_0, H] \,.
\end{equation}
We rewrite
\begin{eqnarray*}
C &=& \left( E^2 + \frac{1}{2}\left(A_1E+EA_1 \right) + \frac{1}{4} A_1^2
\right) - \frac{1}{4} A_1^2 +\frac{1}{2} [A_1, E] + A_0 
\\
&=& \left(E + \frac{1}{2} A_1 \right)^2 - \left( \frac{1}{4} A_1^2
  -\frac{1}{2} [A_1, E] - A_0 \right)\\
&=& \tilde{E}^2-h
\end{eqnarray*}
setting $h=\frac{1}{4} A_1^2 -\frac{1}{2} [A_1, E] - A_0 $ and $\tilde{E}=E
+ \frac{1}{2} A_1$.  Using (\ref{A1A0}) and (\ref{EH}), we compute
\begin{eqnarray*}
 h&=& \frac{1}{4}A_1^2- \frac{1}{2}[N_0,E]+ \frac{1}{2}[H,E]-N_0H\\
 &=& \frac{1}{4}A_1^2- \frac{1}{2}[N_0,E]+
 \frac{1}{2}\left([N_0,E]+[N_0,H]\right) -N_0H\\
 &=& \frac{1}{4}A_1^2-\frac{1}{2}(N_0H+HN_0)
\end{eqnarray*}
such that $[Z,h]=0$.
By inspection, $h^*=h$ as well as $\tilde{E}^*=\tilde{E}$. Moreover, since
$[Z,A_1]=0$, we have $[Z,\tilde{E}]=[Z,E]=i\hbar\mathbf{1}$.
\end{proof}

\begin{lemma}\label{lem:RelFactSpecial} 
  A relativistic constraint of the form $C=(E+g)^2-h$, such that $g^*=g$,
  $h^*=h$, and $[Z,g]=[Z, h]=0$, is deparameterizable by factorization only if
  $[E+g,h]=0$.
\end{lemma}
\begin{proof}
  The factorized version of such a constraint must be of the form
\[
C = (E+ N_0)(E+H) = E^2+(N_0 + H)E + [E, H] + N_0H\ , 
\]
which we compare with
\[
(E+g)^2 - h = E^2+g^2+2gE+[E, g] - h \ .
\]
Taking a commutator with $Z$ and equating the two expressions yields
$g=\frac{1}{2} (N_0 + H)$. Using this result to eliminate $g$ and setting the
two expressions equal gives
\[
[E, H] + N_0H = \frac{1}{4}(N_0+H)^2 + \frac{1}{2} [E, N_0] + \frac{1}{2} [E,
H] - h \ . 
\]
This expression can be rearranged to solve for $h$ in terms of $H$, $N_0$,
and their commutators with $E$
\[
h = \frac{1}{2} \left( [H, E] + [E, N_0] \right) +\frac{1}{2} \left( N_0^2 +
  H^2 + 2N_0H - [N_0, H] \right) - N_0H \ . 
\]
We combine the first two terms using equation~(\ref{eq:RelCond3}) (with $N_1 =
\mathbf{1}$):
\begin{eqnarray*}
h &=&  \frac{1}{2} [N_0, H] + \frac{1}{4} N_0^2 + \frac{1}{4} H^2 -
\frac{1}{2} N_0H - \frac{1}{4} [N_0, H] 
\\
&=& \frac{1}{4} \left( N_0^2 + H^2 - 2N_0H + [N_0, H] \right) = \frac{1}{4}
\left(N_0 - H \right)^2\,. 
\end{eqnarray*}
Now consider the commutator
\begin{eqnarray*}
[E+g, N_0-H] &=& [E+\frac{1}{2}(N_0 + H), N_0-H]
\\
&=& [H, E] + [E, N_0]-\frac{1}{2} [N_0, H] + \frac{1}{2}[H, N_0]
\\
&=& [N_0, H] + [H, N_0] = 0 \ ,
\end{eqnarray*}
where in the final equality we once again used~(\ref{eq:RelCond3}). This
result immediately implies
\[
[E+g, h] =[E+g,\frac{1}{4} \left(N_0 - H \right)^2 ] = 0
\]
as a necessary condition for our constraint to be deparameterizable by
factorization.
\end{proof}

\begin{ex}
We assume that $h=\sqrt{h}^2$ has a square root $\sqrt{h} = \sqrt{h}^*$ in
${\cal A}$. Comparison of the two constraint forms results in
\[
g = \frac{1}{2} (N_0 + H) \ , \quad {\rm and} \ \ \sqrt{h} = \frac{1}{2} (N_0
-H) \ . 
\]
The factorizability condition~(\ref{eq:RelCond3}) now gives
\[
2[E, \sqrt{h} ] = [\sqrt{h}, g ] - [g, \sqrt{h}] \ ,
\]
or
\begin{equation} \label{eq:Factorization}
i\hbar \frac{\partial \sqrt{h}}{\partial Z} = [\sqrt{h}, g] \ .
\end{equation}
For example, in a two-component system with canonical generators $[Z, E]=[q,
p] = i\hbar\mathbf{1}$, setting
\[
\sqrt{h} = p + \frac{1}{2}(q^2 - Z^2) \ , \quad {\rm and} \ \  g = Z(q-Z) \ ,
\]
satisfies equation~(\ref{eq:Factorization}) and leads to the
factorization
\begin{eqnarray*}
C&=&\left( E+Z(q-Z)\right)^2 - \left( p + \frac{1}{2} (q^2-Z^2) \right)^2
\\
&=& \left( E + p + \frac{1}{2}q^2- \frac{3}{2} Z^2 + qZ \right) \left(E -
  \left( p+\frac{1}{2} (q-Z)^2 \right) \right) \ . 
\end{eqnarray*}
\end{ex}

As this example demonstrates, in general a constraint $C\in{\cal A}$ has to be
of a specific form in order for a deparameterization and therefore evolution
with respect to a gauge section to exist. This result showcases the power of
our general approach to quantum dynamical reduction. The restrictions of the
type found in Lemmas~\ref{lem:RelFact}--\ref{lem:RelFactSpecial} have not been
anticipated by the standard method of deparameterization on a fixed Hilbert
space, which treats each specific scenario individually and has mainly been
applied to time-independent systems in which $C=NC_H$, where $N$ and $C_H$
commute. The additional restrictions derived here are the consequence of the
inclusion of time dependence from the outset, as well as the general algebraic
treatment that is not tied to a specific Hilbert-space representation.

\section*{Acknowledgements}

This work was supported in part by NSF grants PHY-1607414 and PHY-1912168.


\end{document}